\documentclass[12pt,letterpaper]{article}
\usepackage[left=1in,top=1in,right=1in,nohead]{geometry}
\usepackage{enumerate}
\usepackage{graphicx}
\usepackage{amsmath}
\usepackage{amsthm}
\usepackage{mathrsfs}
\usepackage{amssymb}
\usepackage{multirow}
\usepackage{color}
\usepackage{bm}
\usepackage{comment}
\usepackage{subfigure}
\usepackage{paralist}
\usepackage{pstricks}
\usepackage{pst-coil,pstricks-add,pst-plot}
\usepackage{jheppub}

\theoremstyle{remark}
\newtheorem{thm}{Theorem}
 
\newtheorem{lem}{Lemma}
\newtheorem{defn}{Definition}

\newcommand{\grad}{\nabla}

\newcommand{\Ocal}{\mathcal{O}}
\newcommand{\eps}{\epsilon}
\newcommand{\epsb}{\bm{\epsilon}}

\newcommand{\Xcal}{\mathcal{X}}
\newcommand{\Acal}{\mathcal{A}}
\newcommand{\Ecal}{\mathcal{E}}
\newcommand{\T}{\mathcal{T}}

\newcommand{\gradt}{\widetilde{\nabla}}

\DeclareMathOperator{\Tr}{Tr}

\def\bea#1\eea{\begin{align}#1\end{align}}
\def\be#1\ee{\begin{equation}#1\end{equation}}

\title{A Bound on Holographic Entanglement Entropy from Inverse Mean Curvature Flow}

\author{Sebastian Fischetti}
\author{and Toby Wiseman}

\affiliation{Theoretical Physics Group, Blackett Laboratory, Imperial College, London SW7 2AZ, UK}

\emailAdd{s.fischetti@imperial.ac.uk}
\emailAdd{t.wiseman@imperial.ac.uk}

\abstract{
Entanglement entropies are notoriously difficult to compute.  Large-$N$ strongly-coupled holographic CFTs are an important exception, where the AdS/CFT dictionary gives the entanglement entropy of a CFT region in terms of the area of an extremal bulk surface anchored to the AdS boundary.  Using this prescription, we show -- for quite general states of (2+1)-dimensional such CFTs -- that the renormalized entanglement entropy of any region of the CFT is bounded from above by a weighted local energy density.  The key ingredient in this construction is the inverse mean curvature (IMC) flow, which we suitably generalize to flows of surfaces anchored to the AdS boundary.  Our bound can then be thought of as a ``subregion'' Penrose inequality in asymptotically locally AdS spacetimes, similar to the Penrose inequalities obtained from IMC flows in asymptotically flat spacetimes.  Combining the result with positivity of relative entropy, we argue that our bound is valid perturbatively in~$1/N$, and conjecture that a restricted version of it holds in any CFT.
}

\begin{document}

\maketitle
\flushbottom

\section{Introduction}
\label{sec:intro}

There is now ample evidence that entanglement entropy plays a crucial role in a broad range of fields, ranging from characterizing quantum phases of matter in condensed matter systems~\cite{KitPre06,LevWen06,HamIon05,GroZha13}; understanding  RG flow and constructing~$c$-theorems~\cite{CasHue04,CasHue06,MyeSin10,MyeSin10b,CasHue12,Sol13}; studying phase transitions in QFTs~\cite{KleKut07,NisTak06,BuiPol08,NakNak11}; and understanding the emergence of classical spacetime from quantum gravity~\cite{Van09,Van10,MalSus13}.  Unfortunately, tractable methods for explicitly calculating entanglement entropy remain elusive.  A notable exception, however, is  holography in the context of AdS/CFT~\cite{Mal97,Wit98a,GubKle98}.  Specifically, in the large-$N$, large-$\lambda$ limit, the entanglement entropy of some spatial region~$B$ in a holographic CFT is given by the HRT formula~\cite{RyuTak06,RyuTak06b,HubRan07,LewMal13,DonLew16}:
\be
\label{eq:HRT}
S_B = \min_{X_B \sim B} \frac{\mathrm{Area}[X_B]}{4G_N},
\ee
where the minimization is performed over bulk extremal surfaces~$X_B$, and here~$\sim$ means ``homologous to''.

The HRT formula has been remarkably useful in understanding properties of the bulk gravitational theory from properties of~$S_B$; see e.g.~\cite{LasRab14,LasLin16,FauGui13,LasMcD13,SwiVan14,%
LinMar14,JafSuh14,JafLew15,DonHar16}.  However, it suffers a drawback: any computation of~$S_B$ relies quite explicitly on knowing the bulk geometry.  Thus equation~\eqref{eq:HRT} does not immediately take on any transparent interpretation in terms of field theory data, and does not substantially address the converse question of understanding entanglement entropy in purely field theoretic terms.

Our purpose in this paper is to address this deficiency, at least in the holographic context.  To that end, recall that in any state on~$B$ it is possible to define a state-dependent operator~$H_B$ -- the modular Hamiltonian -- so that the entanglement entropy can be written as~$S_B = \langle H_B \rangle$.  Like~$S_B$,~$H_B$ is in general very complicated and thus rarely known explicitly\footnote{Holographic considerations using the HRT formula give a bulk interpretation of~$H_B$~in terms of the modular Hamiltonian and area operator in the bulk~\cite{JafSuh14,JafLew15}, but this interpretation exchanges one intractable operator for another.}.  A natural attempt to alleviate this difficulty is to relax the equality~$S_B = \langle H_B \rangle$ to a bound
\be
\label{eq:roughbound}
S_B \leq \langle \Ecal_B \rangle
\ee
for some appropriately ``nice'' operator~$\Ecal_B$.  Ideally, such an operator should reduce to the modular Hamiltonian~$H_B$ in cases where the latter is known.

Of course, as stated the bound~\eqref{eq:roughbound} is trivial, because~$S_B$ is divergent and thus~$\langle \Ecal_B \rangle$ must be divergent as well.  Let us therefore be more precise, restricting our attention to~(2+1)-dimensional field theories.  Recall that in a continuum QFT, the entanglement entropy~$S_B$ is divergent due to short-distance correlations across the entanglement surface~$\partial B$.  Introducing a UV regulator~$\eps$, the leading piece of this divergence takes the ``area law'' form~\cite{BomKou86,Sre93}
\be
\label{eq:arealaw}
S_B = a_{-1} \, \frac{L}{\eps} + S^\mathrm{ren}_B + \cdots,
\ee
where~$a_{-1}$ is some state-independent constant and~$L$ is the length of~$\partial B$.  In a general QFT, the suppressed terms~$\cdots$ may in fact contain sub-leading state-dependent divergent terms; see e.g.~\cite{MarWal16}.  But when the area law term in~\eqref{eq:arealaw} is the only divergent piece of~$S_B$ -- such as in a large-$N$ holographic CFT deformed by appropriate operators -- all state dependence is contained in the finite piece~$S^\mathrm{ren}_B$, termed the \textit{renormalized} entanglement entropy\footnote{Note that as defined here the renormalized entanglement entropy appears UV-cutoff dependent; this will be remedied when we define it more precisely below. \label{foot:renscheme}} (sometimes also called the ``universal'' contribution to the entropy).

Then our main result is that -- under certain assumptions to be made more precise later -- in any state of a~(2+1)-dimensional holographic CFT, the renormalized entanglement entropy obeys
\be
\label{eq:sketchybound}
S^\mathrm{ren}_B \leq \left\langle \Ecal_B \right\rangle - 8\pi^2 c_\mathrm{eff} \chi_B,
\ee
where~$c_\mathrm{eff} \equiv \ell^2/(16\pi G_N) \sim N^2$ is the effective central charge of the CFT,~$\chi_B$ is the Euler characteristic of~$B$, and
\be
\Ecal_B \equiv 2\pi \int_B \omega \varepsilon.
\ee
Here~$\varepsilon$ is the energy density on~$B$ and~$\omega$ is a state-dependent (but c-number) positive weighting function which can be computed explicitly from the dual geometry; thus like~$H_B$,~$\Ecal_B$ is a state-dependent operator.  We emphasize that~\eqref{eq:sketchybound} holds not just in the limit of a classical bulk, but for any state which is a finite-order perturbation (say, in~$1/N$) of a state dual to Einstein gravity.  In such a case, the weighting function~$\omega$ is calculated only from the classical geometry, neglecting perturbative corrections.  In the special case when~$B$ is a disk on Minkowski space and the CFT is in vacuum, the operator~$\Ecal_B$ corresponds precisely to~$H_B$ and the inequality in~\eqref{eq:sketchybound} becomes saturated.  Thus one can think of~$\Ecal_B$ as a tractable generalization of~$H_B$, and the properties of~$\omega$ allow us to make bulk-independent statements about~$S^\mathrm{ren}_B$; for example, whenever~$\langle\varepsilon\rangle$ is non-positive on~$B$, it follows that~$S^\mathrm{ren}_B$ is bounded above by~$-8\pi^2 c_\mathrm{eff} \chi_B$.

The key ingredient in proving~\eqref{eq:sketchybound} is a geometric flow called the inverse mean curvature (IMC) flow, first studied in~\cite{Ger73,JanWal77} in the context of proving Penrose inequalities in asymptotically flat spacetimes (that is, lower bounds on the ADM mass in terms of the area of any horizons).  It was shown in~\cite{Gib98} that similar inequalities can be obtained from IMC flows in black hole spacetimes with homogeneous slicings, but Lee and Neves have shown that in general asymptotically AdS spacetimes a similar approach yields much weaker bounds than in the asymptotically flat case~\cite{Nev10,LeeNev15}.  Our approach here can be thought of as a strengthening of the results of~\cite{LeeNev15}, which we do by introducing a new IMC flow in which the flow surfaces are anchored to the AdS boundary.  The bound~\eqref{eq:sketchybound} can thus be thought of as a generalized Penrose inequality applied to subregions of the AdS boundary.

The essence of our approach is as follows.  By the homology constraint~\cite{HeaTak07} in the HRT formula, there exists a spatial slice ``enclosed'' by~$B$ and~$X_B$.  The IMC flow defines a preferred foliation of this slice into two-dimensional surfaces that flow from~$X_B$ to~$B$, as shown in Figure~\ref{fig:IMCsketch}.  Under this flow, the so-called reduced Hawking mass~$I_H$ obeys a monotonicity property, yielding precisely the bound~\eqref{eq:sketchybound}.  It is especially worth noting that to obtain this result, we must perform the flow over a \textit{maximal-volume} spatial slice enclosed by~$X_B$ and~$B$, the volume of which has recently been conjectured to be dual to the complexity of the state on~$B$~\cite{BenCar16} (see also~\cite{Sus14,Sus14b,StaSus14,BroRob15,BroRob15b,BroSus16,ChaMar16} for related ideas).

\begin{figure}[t]
\centering
\includegraphics[width=0.55\textwidth,page=1]{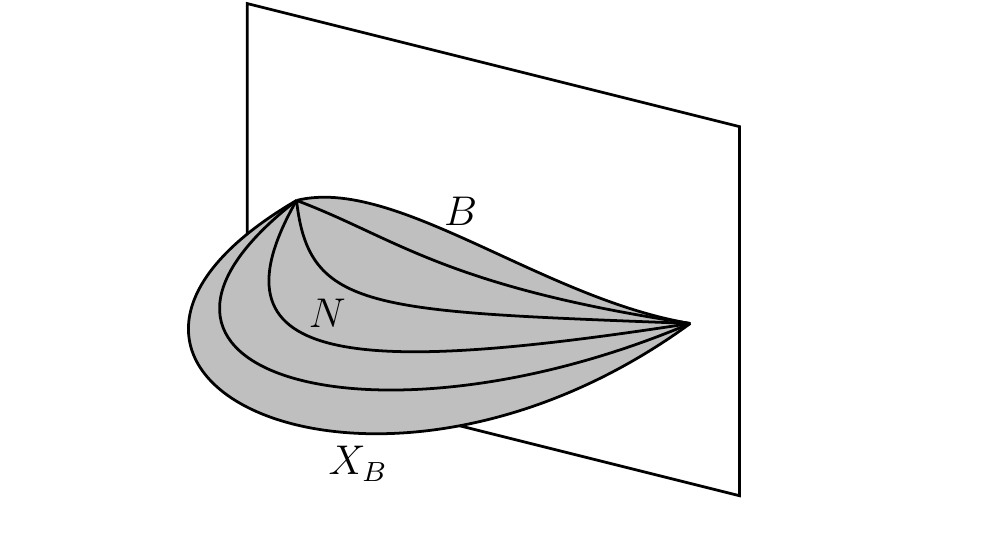}
\caption{$B$ is a spatial region of a (2+1)-dimensional CFT on the boundary, and~$X_B$ is the corresponding HRT surface.  By the homology constraint there exists a spatial slice~$N$ (shaded gray) with boundary given by~$X_B \cup B$; when~$N$ is chosen to be a maximal-volume slice, the IMC flow on~$N$ from~$X_B$ to~$B$ (shown as a foliation of~$N$) produces the bound~\eqref{eq:sketchybound}.}
\label{fig:IMCsketch}
\end{figure}

This paper is organized as follows.  In Section~\ref{sec:IMCF}, we introduce the IMC flow and review the monotonicity property of the reduced Hawking mass, and then derive a generalization to the case of boundary-anchored flow surfaces.  In Section~\ref{sec:bounds} we compute the asymptotic behavior of the reduced Hawking mass and apply the monotonicity of the flow to derive~\eqref{eq:sketchybound}, stating our assumptions precisely in the process.  The reader only interested in the final result may skip the calculations in Section~\ref{sec:bounds} and just read the final statement of Theorem~\ref{thm:bound}.  In Section~\ref{sec:perturbative} we note that~$\Ecal_B$ and~$H_B$ agree in the case where the latter is known explicitly (i.e.~in the vacuum state on Minkowski space and when~$B$ is a disk).  Combined with positivity of the relative entropy, we use this feature to argue that our bound should hold for any state which is a perturbative correction of a classical one.  We conclude in Section~\ref{sec:disc} with a discussion of open directions and generalizations, and possible connections to complexity.  Appendices~\ref{app:Kconformal},~\ref{app:massreg},~\ref{app:intevolution}, and~\ref{app:pureAdS} contain useful results we will invoke throughout the text.

\subsection*{Preliminaries}

Unless otherwise specified, our conventions will be as follows.  We will consider a~(3+1)-dimensional asymptotically locally AdS spacetime~$(M,g_{ab})$ obeying Einstein's equation with negative cosmological constant:
\be
\label{eq:Einstein}
R_{ab} = -\frac{3}{\ell^2} \, g_{ab} + 8\pi G_N \left(T_{ab} - \frac{1}{2} T g_{ab} \right),
\ee
with~$\ell$ the AdS length.  We will require that the matter obey the weak energy condition:
\be
T_{ab} \xi^a \xi^b \geq 0 \mbox{ for all timelike } \xi^a.
\ee
We use~$\Omega$ to denote a defining function of some conformal frame; that is,~$\Omega$ is a scalar field such that the conformally rescaled metric~$\tilde{g}_{ab} = \Omega^2 g_{ab}$ is regular on (at least some portion of) the asymptotic boundary, which we will always refer to as~$\partial M$.  $Z$ will refer to a Fefferman-Graham defining function~\cite{FefGra85,GraWit99}, which has the properties that
\be
\label{eq:FGZ}
|\gradt Z|^2 \equiv \tilde{g}^{ab} (\gradt_a Z) (\gradt_b Z) = \frac{1}{\ell^2}
\ee
in a neighborhood of the boundary, and that~$\partial M$ is a totally geodesic surface\footnote{Recall that a totally geodesic surface~$N$ embedded in~$M$ is one such that any geodesic of~$N$ is also a geodesic of~$M$, and thus the extrinsic curvature of~$N$ vanishes.} with respect to~$\tilde{g}_{ab}$ (see~\cite{FisKel12} for a review).

We will often be concerned with boundary-anchored surfaces, defined as:
\begin{defn}
\label{def:bndryanchored}
A \textit{boundary-anchored surface}~$\Sigma$ in an asymptotically locally AdS spacetime~$(M,g_{ab})$ is a surface with asymptotic boundary~$\partial \Sigma \subset \partial M$ such that (i) there exists a conformal frame~$\tilde{g}_{ab} = \Omega^2 g_{ab}$ which is regular in a neighborhood of~$\partial \Sigma$, and (ii)~$\Sigma$ is~$C^1$ with respect to~$\tilde{g}_{ab}$ in a neighborhood of~$\partial \Sigma$ (i.e.~the extrinsic curvature of~$\Sigma$ with respect to~$\tilde{g}_{ab}$ is finite).
\end{defn}

We will introduce a maximal-volume time slice~$N$ of this spacetime whose induced metric, intrinsic curvature, and extrinsic curvature will be denoted by~$h_{ab}$,~$^N \! R_{ab}$, and $^N \! K_{ab}$, respectively, and whose unit (timelike) normal will be~$t^a$.  The fact that~$N$ is a maximal-volume slice implies that its mean curvature is zero:~$\, ^N \! K \equiv h^{ab} \, ^N \! K_{ab} = 0$.  On~$N$, we will also consider a one-parameter family of two-dimensional surfaces~$\Sigma_\tau$; to avoid clutter, we will often avoid writing the parameter~$\tau$ explicitly.  The induced metric and intrinsic and extrinsic curvatures (in~$N$) of the surfaces~$\Sigma_\tau$ will be denoted by~$\sigma_{ab}$,~$^\Sigma \! R_{ab}$, and~$^\Sigma \! K_{ab}$, respectively, while the unit (spacelike) normals to the~$\Sigma_\tau$ in~$N$ will be denoted by~$n^a$.

Finally, the Gauss-Codazzi equations relate the intrinsic curvature of~$N$ to its extrinsic curvature as
\be
\label{eq:GaussCodazzi}
^N \! R = -\frac{6}{\ell^2} + \, ^N \! K_{ab} \, ^N \! K^{ab} + 16\pi G_N T_{ab} t^a t^b,
\ee
where we have explicitly used the fact that the mean curvature of~$N$ vanishes.

\section{The Inverse Mean Curvature Flow}
\label{sec:IMCF}

Our discussion of renormalized entropy has so far been relatively vague: as noted in footnote~\ref{foot:renscheme},~$S^\mathrm{ren}_B$ is in general cutoff-dependent (as can be easily seen by considering a rescaling~$\eps \to \eps(1 + a \eps)$ of the UV cutoff).  Let us therefore begin by clarifying what we mean by~$S^\mathrm{ren}_B$; we will find that a more precise formulation leads us quite naturally to introduce the inverse mean curvature flow advertised above.

\subsection{Renormalized Area}
\label{subsec:renarea}

By the HRT formula, the entanglement entropy of some boundary region~$B$ is given by the area of the minimal-area extremal surface~$X_B$ anchored to~$\partial B$.  Because~$X_B$ reaches the asymptotic boundary, its area is infinite, and thus needs to be regulated.  A particularly useful choice of regulation consists of introducing a Fefferman-Graham defining function~$Z$ and excising the portion of~$X_B$ at~$Z < \eps/\ell$ for some small cutoff~$\eps$, as shown in Figure~\ref{fig:regulatedarea}.  If the components of the bulk stress tensor in any orthonormal frame fall off like~$T_{\hat{\mu}\hat{\nu}} = \Ocal(Z^2)$ near the conformal boundary, this regulated area has the well-known expansion
\be
\label{eq:Aexpansion}
A[X_B^\eps] = \ell^2 \frac{L}{\eps} + \Acal[X_B] + \Ocal(\eps),
\ee
where~$L$ is the length of~$\partial B$ in the conformal frame associated to~$Z$.  As shown by Graham and Witten~\cite{GraWit99}, because the cutoff~$\eps$ is defined by a Fefferman-Graham defining function~$Z$, the constant piece~$\Acal[X_B]$ is independent of the conformal frame associated to~$Z$.  Thus~$\Acal[X_B]$ provides a frame- and cutoff-independent definition of the holographic renormalized entanglement entropy:~$S^\mathrm{ren}_B = \Acal[X_B]/4G_N$.

\begin{figure}[t]
\centering
\includegraphics[width=0.35\textwidth,page=2]{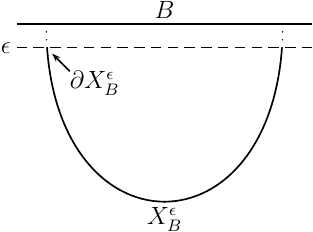}
\caption{A boundary-anchored extremal surface~$X_B$ has infinite area.  By introducing a cutoff at~$Z = \eps/\ell$ (shown as a dashed line), the~$\eps$-independent piece of~$A[X_B^\eps]$ is a conformal invariant and defines the renormalized area~$\Acal[X_B]$.  This renormalized area can alternatively be defined via the subtraction of covariant counterterms on the regulated boundary~$\partial X_B^\eps$.}
\label{fig:regulatedarea}
\end{figure}

In fact, we may define the renormalized area via an appropriate subtraction:
\be
\Acal[X_B] = \lim_{\eps \to 0} \left(A[X_B^\eps] + A_\mathrm{ct}[X_B^\eps] \right),
\ee
where~$A_\mathrm{ct}[X_B^\eps] = -\ell^2 L/\eps + \Ocal(\eps)$.  If one wishes to define~$X_B$ as a stationary point of the functional~$\Acal[X_B]$, then the counterterm~$A_\mathrm{ct}[X_B^\eps]$ must be chosen to impose Dirichlet boundary conditions on~$X_B$.  In such a case, its form is fixed to simply be the proper length of~$\partial X_B^\eps$~\cite{TayWoo16}.  However, here we are only interested in on-shell evaluations of~$\Acal[X_B]$ -- that is, we wish to compute~$\Acal[X_B]$ on an \textit{a priori} specified surface~$X_B$ -- so that we may choose any~$A_\mathrm{ct}[X_B^\eps]$ without worrying about its variations (in an abuse of terminology, we will continue to call~$A_\mathrm{ct}$ a counterterm).  To that end, note that as found in~\cite{AleMaz08} (and reviewed in Lemma~\ref{lem:geocurvature} of Appendix~\ref{app:Kconformal}), we may take~$A_\mathrm{ct}[X_B^\eps]$ to be
\be
A_\mathrm{ct}[X_B^\eps] = \ell^2 \int_{\partial X_B^\eps} k_g = -\ell^2 \frac{L}{\eps} + \Ocal(\eps),
\ee
where~$k_g$ is the geodesic curvature of~$\partial X_B^\eps$ in~$X_B$\footnote{That is, if~$a^a$ is the geodesic acceleration of~$\partial X_B^\eps$ and~$v^a$ is the unit outward-pointing normal to~$\partial X_B^\eps$ in~$X_B$, then~$k_g = v_a a^a$. \label{foot:kg}}.  This choice of counterterm has the advantage that the Gauss-Bonnet theorem relates it to the scalar curvature of~$X_B^\eps$ as
\be
\label{eq:GaussBonnet}
\int_{X_B^\eps} \, ^{X_B} \! R = 4\pi \chi_{X_B} + 2 \int_{\partial X_B^\eps} k_g,
\ee
where~$^{X_B} \! R$ and~$\chi_{X_B}$ are the Ricci scalar and Euler characteristic of~$X_B$.  We can thus express the renormalized area of~$X_B$ directly as an integral over all of~$X_B$, without any need for a UV cutoff at all:
\be
\Acal[X_B] = \ell^2 \left[\frac{1}{4}\int_{X_B} \left(2 \, ^{X_B} \! R + \frac{4}{\ell^2} \right) - 2\pi \chi_{X_B}\right].
\ee

The integral appearing in the expression above bears a striking resemblance to the so-called Hawking mass, defined on any surface~$\Sigma$ as~$m_H[\Sigma] = \sqrt{A[\Sigma]} \, I_H[\Sigma]$, with
\be
\label{eq:redHawkingmass}
I_H[\Sigma] \equiv \int_\Sigma \left(2 \, ^\Sigma \! R + \frac{4}{\ell^2} - \, ^\Sigma \! K^2 \right).
\ee
We will call~$I_H[\Sigma]$ the reduced Hawking mass of~$\Sigma$.  Since~$X_B$ is an extremal surface (and thus~$^{X_B} \! K = 0$), we therefore find that the renormalized area~$\Acal[X_B]$ can be written as
\be
\label{eq:ArenI}
\Acal[X_B] = \ell^2\left(\frac{1}{4} I_H[X_B] - 2\pi \chi_{X_B}\right).
\ee
It is important to note that the factor of~$^\Sigma \! K^2$ in~$I_H[\Sigma]$ is not a trivial addition: as we show in Appendix~\ref{app:massreg}, it renders~$I_H[\Sigma]$ finite for \textit{any} boundary-anchored surface~$\Sigma$ lying on an extremal time slice.

We may now immediately obtain a partial bound on~$\Acal[X_B]$ and thus~$S^\mathrm{ren}_B$, reproducing one first obtained in~\cite{AleMaz08}.  To that end, consider the case when the bulk is pure AdS and when~$B$ lies on a static time slice of the boundary, and consider an extremal surface~$\Xcal_B$ anchored on~$\partial B$ with the same topology as~$B$.  Note that~$\Xcal_B$ need not be the same as the HRT surface~$X_B$, but by the HRT formula we must have that~$\Acal[\Xcal_B] \geq \Acal[X_B]$.  Then in Appendix~\ref{app:massreg} we show that
\be
\label{eq:IAdS}
I_H[\Xcal_B] = -2 \int_{\Xcal_B} \left|^{\Xcal_B} \! K_{ab}\right|^2 \leq 0,
\ee
and thus by~\eqref{eq:ArenI} we find
\be
\label{eq:partialbound}
S^\mathrm{ren}_B = \frac{\Acal[X_B]}{4G_N} \leq \frac{\Acal[\Xcal_B]}{4G_N} = \frac{\ell^2}{4G_N} \left(\frac{1}{4} I_H[\Xcal_B] - 2\pi\chi_B \right) \leq -8\pi^2 c_\mathrm{eff} \chi_B,
\ee
where we have introduced the effective central charge~$c_\mathrm{eff} \equiv \ell^2/16\pi G_N$.  This result bounds the vacuum renormalized entanglement entropy of arbitrary regions of a holographic CFT$_3$ on Minkowski space.  In fact, it is easy to check that taking~$B$ to be a disk saturates the inequalities in~\eqref{eq:partialbound} (recall that~$\chi = 1$ for a disk), so we conclude that disks maximize~$S^\mathrm{ren}_{\chi_B = 1}$, as was found globally in~\cite{AleMaz08} and locally in~\cite{Mez14,AllMez14,FauLei15}.

Can this bound be generalized to arbitrary states and boundary metric?  In general,~$I_H[\Xcal_B]$ does not have a definite sign, and thus the approach above appears unfruitful.  However, the Hawking mass~$m_H$ has the special feature that it is monotonic along the inverse mean curvature flow.  This property, which we will now review, can be generalized to extend~\eqref{eq:partialbound} into a much more powerful bound.

\subsection{IMC Flow for Compact Surfaces}
\label{subsec:IMCcompact}

Let us first review the IMC flow in the case of compact surfaces without boundary, where the discussion is more streamlined.  Consider a maximal-volume time slice~$N$ and a one-parameter family~$\Sigma_\tau$ of compact two-dimensional surfaces without boundary (refer to the end of Section~\ref{sec:intro} for notation).  The IMC flow is then defined by requiring that the~$\Sigma_\tau$ move along their unit normals with speed equal to their inverse mean curvature.  Explicitly, we require that
\be
\label{eq:flow}
\pounds_{\hat{n}} \tau = 1,
\ee
where~$\hat{n}^a \equiv \, ^\Sigma \! K^{-1} n^a$ is the normal to the flow surfaces with magnitude~$|\hat{n}^a| = \, ^\Sigma \! K^{-1}$, with~$^\Sigma K \equiv \sigma^{ab} \, ^\Sigma K_{ab}$ the mean curvature of the~$\Sigma_\tau$ in~$N$.  One may alternatively think of the~$\Sigma_\tau$ as level sets of the scalar~$\tau$ (the ``flow time'') on~$N$, in which case the flow equation~\eqref{eq:flow} can be re-expressed in terms of~$\tau$ as
\be
\label{eq:flowpde}
\grad_a\left(\frac{\grad^a \tau}{|\grad\tau|}\right) = |\grad\tau|,
\ee
with~$\grad_a$ the covariant derivative operator on~$N$.

By~\eqref{eq:flow}, evolution along the flow is given by~$\Sigma_{\tau'} = \Phi^{\tau' - \tau} \Sigma_\tau$, where~$\Phi^\tau$ is the flow along integral curves of~$\hat{n}^a$ by a time~$\tau$.  It thus follows that the rate of change of any flow quantities with respect to~$\tau$ can be computed by taking Lie derivatives along~$\hat{n}^a$.  For instance, the rate of change of the induced metric~$\sigma_{ab}$ is
\be
\label{eq:sigmadot}
\dot{\sigma}_{ab} \equiv \pounds_{\hat{n}} \sigma_{ab} = \frac{1}{^\Sigma \! K} \pounds_n \sigma_{ab} = \frac{2}{^\Sigma \! K} \, ^\Sigma \! K_{ab},
\ee
while the Gauss-Codazzi equations yield
\be
\label{eq:GaussCodazzi2}
^\Sigma \! \dot{K} = \frac{1}{2 \, ^\Sigma \! K} \left(- \, ^N \! R + \, ^\Sigma \! R - \, ^\Sigma \! K^2 - |\, ^\Sigma \! K_{ab}|^2 - 2 \left|D_a \ln \, ^\Sigma \! K \right|^2 + 2 D^2 \ln \, ^\Sigma \! K\right),
\ee
where~$D_a$ is the covariant derivative operator on~$\Sigma$.

Of particular interest are integrated quantities: we show in Appendix~\ref{app:intevolution} that along an IMC flow of compact surfaces without boundary, the rate of change of any integral~$F[\Sigma] = \int_\Sigma f$ for scalar~$f$ is given by
\be
\label{eq:Fdotcompact}
\dot{F}[\Sigma] = \int_\Sigma \left(\dot{f} + f\right).
\ee
In the special case~$f = 1$, we obtain~$\dot{A}[\Sigma] = A[\Sigma]$, and thus the area of the flow surfaces grows exponentially in flow time:
\be
\label{eq:Atau}
A[\Sigma_\tau] = e^\tau A[\Sigma_0].
\ee
Likewise, we may use the Gauss-Bonnet theorem~\eqref{eq:GaussBonnet} (with no boundary term) to write the reduced Hawking mass of~$\Sigma$ as
\be
I_H[\Sigma] = 8\pi \chi_\Sigma + \int_{\Sigma} \left(\frac{4}{\ell^2} - \, ^\Sigma K^2 \right);
\ee
then taking~$f = 4/\ell^2 - \, ^\Sigma \! K^2$, we obtain~$F[\Sigma] = I_H[\Sigma] - 8\pi\chi_\Sigma$.  As long as the flow is smooth (so that the topology of the~$\Sigma$ remains unchanged), we may use~\eqref{eq:GaussCodazzi},~\eqref{eq:GaussCodazzi2}, and the Gauss-Bonnet theorem to obtain
\begin{subequations}
\label{eq:Idot}
\bea
\dot{I}_H[\Sigma] + \frac{1}{2} I_H[\Sigma] &= \int_\Sigma \left(\left| ^N \! K_{ab} \right|^2 + \left| ^\Sigma \widehat{K}_{ab} \right|^2 + 2 \left|D_a \ln \, ^\Sigma \! K \right|^2 + 16\pi G_N T_{ab} t^a t^b \right), \label{subeq:Idotint} \\
	&\geq 0,
\eea
\end{subequations}
where we have integrated by parts and defined the trace-free extrinsic curvature~$^\Sigma \widehat{K}_{ab} \equiv \, ^\Sigma K_{ab} - \, ^\Sigma \! K \sigma_{ab}/2$ of~$\Sigma$, and to obtain the inequality we exploited the fact that~$T_{ab}$ obeys the weak energy condition.  This proves the advertised monotonicity: the above implies that~$e^{\tau/2} I_H[\Sigma]$ is non-decreasing along the flow.  In terms of the Hawking mass~$m_H[\Sigma] = \sqrt{A[\Sigma]} \, I_H[\Sigma]$, we note from~\eqref{eq:Atau} that~$e^{\tau/2} \propto \sqrt{A}$, and thus~$m_H[\Sigma]$ is monotonic along the flow as well.

We have thus established the monotonicity of the Hawking mass under a smooth IMC flow of compact surfaces.  However, it is important to note that IMC flows are typically not smooth; they may encounter singularities or cusps past which the flow fails to proceed smoothly (for example, when~$^\Sigma \! K$ vanishes).  In such a case, the analysis above fails.  However, it is possible to generalize it: in the asymptotically flat case, Huisken and Ilmanen~\cite{HuiIlm01} showed that a regulated version of~\eqref{eq:flowpde} can be used to define a ``weak'' IMC flow wherein the flow surfaces are allowed to ``jump'' over singularities of the flow while preserving monotonicity of the Hawking mass.  These results were later generalized to the asymptotically hyperbolic context by Lee and Neves~\cite{LeeNev15}, so the monotonicity result outlined here holds even when a smooth flow does not exist, as long as the flow surfaces do not change their topology.  Thus starting the flow on an apparent horizon and flowing out to the asymptotic boundary, monotonicity of~$m_H[\Sigma]$ yields a lower bound on the asymptotic Hawking mass in terms of the area of the apparent horizon.  In an asymptotically flat spacetime,~$m_H[\Sigma]$ asymptotically approaches the ADM mass, and thus one obtains a Penrose inequality~\cite{JanWal77,HuiIlm01}.  On the other hand, in the asymptotically AdS case~$m_H[\Sigma]$ approaches the integrated energy density of the boundary in some conformal frame defined by the flow, which in general does not coincide with the mass of the spacetime except for highly symmetric configurations~\cite{Gib98,Nev10,LeeNev15}; thus the resulting inequality is substantially weaker.

\subsection{IMC Flow for Surfaces with Boundary}
\label{subsec:IMCbndry}

Our main purpose now is to generalize the above monotonicity result to IMC flows of surfaces anchored to the AdS boundary.  Of course, this process first requires defining what is meant by a ``boundary-anchored flow''.  Moreover, because the flow surfaces have a boundary, we must then ensure that any boundary terms picked up by~$\dot{I}_H$ have the correct sign to preserve monotonicity of~$e^{\tau/2} I_H$.  With an appropriate definition of boundary-anchored flow, we will in fact show that these terms vanish, so that~\eqref{eq:Idot} continues to hold.

\subsubsection{Boundary-Anchored Flows}
\label{subsec:bndryflow}

\begin{figure}[t]
\centering
\subfigure[]{
\includegraphics[width=0.3\textwidth,page=3]{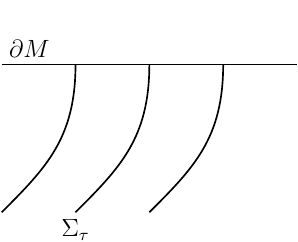}
\label{subfig:movingflow}
}
\hspace{0.1cm}
\subfigure[]{
\includegraphics[width=0.3\textwidth,page=4]{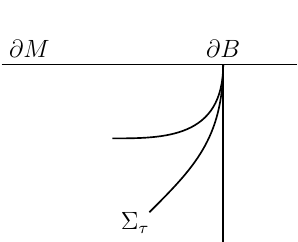}
\label{subfig:badflow}
}
\hspace{0.1cm}
\subfigure[]{
\includegraphics[width=0.3\textwidth,page=5]{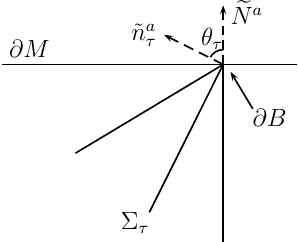}
\label{subfig:goodflow}
}
\caption{Solutions to the IMC flow equation may in principle exhibit different types of behaviors when the IMC surfaces are anchored to the asymptotic boundary~$\partial N$.  For example, the surfaces may be asymptotically extremal and ``move'' along the boundary, as in~\subref{subfig:movingflow}; they may remain anchored on the same curve~$\partial B$ at the boundary but always intersect the boundary at the same angle, as in~\subref{subfig:badflow}; or they may remain anchored on the same curve~$\partial B$ while the angle~$\theta_\tau$ of their intersection with the boundary changes with~$\tau$, as in~\subref{subfig:goodflow}.  Here we will only consider case~\subref{subfig:goodflow} by imposing appropriate boundary conditions on the flow.}
\label{fig:possibleIMCs}
\end{figure}

How should one think of a boundary-anchored flow?  Specifically, how do solutions to the flow equation~\eqref{eq:flow} behave when the flow surfaces reach the asymptotic boundary?  To address these questions, in Figure~\ref{fig:possibleIMCs} we sketch the near-boundary behavior of some flows which are in principle possible solutions to the IMC flow equation.  Now, recall that we are ultimately interested in IMC flows that start on an extremal surface and asymptote to a specified boundary region~$B$.  Thus the only boundary-anchored flows in which we are interested are those that behave as shown in Figure~\ref{subfig:goodflow}: that is, the flow surfaces should all be anchored to the same curve~$\partial B$ on the AdS boundary, while the angle of their intersection with the boundary should change smoothly with~$\tau$.  More precisely, consider a conformal compactification~$\tilde{h}_{ab}$ of the spacetime which is regular at~$\partial B \subset \partial N$; also define~$\widetilde{N}^a$ and~$\tilde{n}^a_\tau$ as the unit (with respect to~$\tilde{h}_{ab}$) normals to~$\partial N$ and~$\Sigma_\tau$, respectively.  Then the flow shown in Figure~\ref{subfig:goodflow} has the property that~$\theta_\tau$, defined via
\be
\label{eq:thetatau}
\cos\theta_\tau \equiv \widetilde{N}_a \tilde{n}^a_\tau,
\ee
is a smooth always-varying function of~$\tau$.

It immediately follows (either from the above definition or from Figure~\ref{subfig:goodflow}) that~$\tau$ must be multi-valued at~$\partial B$, and that this multi-valued singularity is captured by
\be
\label{eq:tausing}
\tau|_{\partial B} = f(\widetilde{N}_a \tilde{n}^a_\tau)
\ee
for a (nonzero) smooth function~$f$.  It is therefore natural to take as a \textit{definition} of a boundary-anchored IMC flow the requirement that near~$\partial B$,~$\tau$ exhibit the singularity structure~\eqref{eq:tausing} on top of any single-valued behavior.  We therefore build these ingredients into the following definition:

\begin{defn}
\label{def:boundaryIMC}
Let~$(N,h_{ab})$ be an asymptotically locally hyperbolic three-dimensional space and let~$\tilde{h}_{ab}$ be a conformal completion of it.  A \textit{regular boundary-anchored IMC flow} is a one-parameter family of two-dimensional surfaces~$\Sigma_\tau$ anchored to a curve~$\partial B \subset \partial N$ such that (i) the flow equation~\eqref{eq:flowpde} is obeyed and (ii) near~$\partial B$ in~$N$,~$\tau$ can be expressed as
\be
\label{eq:boundaryIMCdef}
\tau = f(\tilde{\theta}) + \tau_\mathrm{reg}.
\ee
Here~$f$ is a smooth nonzero function of its argument when~$\tilde{\theta} \in (0,\pi/2]$;~$\tilde{\theta}$ is an arbitrary scalar field whose restrictions to each~$\Sigma_\tau$ are~$C^1$ and which approach~$\theta_\tau$ (defined in~\eqref{eq:thetatau}) at~$\partial B$; and~$\tau_\mathrm{reg}$ is single-valued and~$|\gradt \tau_\mathrm{reg}|$ finite at~$\partial B$ .
\end{defn}

This definition provides a boundary condition guaranteeing that near the asymptotic boundary, the IMC flow looks like Figure~\ref{subfig:goodflow}.  To show that this definition is compatible with the flow equation~\eqref{eq:flowpde}, and thus that there do in fact exist regular boundary-anchored IMC flows as defined above, let us construct a simple example in pure AdS.  We restrict our attention to the Poincar\'e patch, in which static spatial slices of constant~$t$ are maximal-volume slices whose metric we will write in the form
\be
\label{eq:PoincareAdS}
ds^2_N = \frac{\ell^2}{z^2} \left( dz^2 + dx^2 + dy^2 \right) = \frac{\ell^2}{\xi^2 \sin^2\theta} \left(d\xi^2 + \xi^2 \, d\theta^2 + dy^2\right),
\ee
where~$z = \xi \sin\theta$ and~$x = \xi \cos\theta$.

We will consider a flow anchored to the line~$x = z = 0$.  As shown in Appendix~\ref{app:pureAdS}, the IMC flow equation~\eqref{eq:flowpde} is solved by
\be
\label{eq:halfplaneflow}
\tau = \ln\csc^2\theta,
\ee
corresponding to an IMC flow starting on the extremal surface~$x = 0$ and asymptoting to the half-plane~$x > 0$,~$z = 0$.  This flow takes the form~\eqref{eq:boundaryIMCdef} with~$\tilde{\theta} = \theta$, and is therefore an example of a regular boundary-anchored IMC flow.  In the proof of Theorem~\ref{thm:monotonicity} below, we will construct the asymptotics of such flows much more generally.

\subsubsection{Intuition: Vanishing Boundary Terms}
\label{subsec:boundaryint}

We now need to examine the behavior of the reduced Hawking mass on boundary-anchored flows.  As mentioned above,~$\dot{I}_H$ will pick up boundary terms, which in principle could spoil the monotonicity property of~$e^{\tau/2} I_H$.  To gain some intuition for the role these boundary terms might play, let us again consider the pure AdS flow~\eqref{eq:halfplaneflow}.  It is straightforward to compute the intrinsic and extrinsic curvature of the flow surfaces:
\be
^\Sigma \! K = \frac{2}{\ell}\sqrt{1-e^{-\tau}}, \qquad ^\Sigma \! R = -e^{-\tau}\frac{2}{\ell^2},
\ee
and therefore we find that the integrand of the reduced Hawking mass~$I_H$ vanishes.  It is also easy to see that the integrand in the right-hand side of~\eqref{subeq:Idotint} vanishes as well, from which we conclude that~$\dot{I}_H$ cannot receive a contribution from any potential boundary terms.

Next, consider a general flow anchored to a curve~$\partial B$ on the boundary of a general asymptotically locally AdS spacetime.  It is always possible (at least locally) to choose a conformal frame in which~$\partial B$ is geodesic; let~$z$ be the associated Fefferman-Graham coordinate on the extremal time slice~$N$.  In this frame, we choose the boundary coordinate~$y$ to be the proper distance along~$\partial B$ and the boundary coordinate~$x$ to be a Gaussian normal coordinate fired off of~$\partial B$.  With this choice of coordinates, the metric on~$N$ near~$\partial B$ takes the form
\be
\label{eq:nearboundary}
ds^2_N = \frac{\ell^2}{z^2}\left[dz^2 + dx^2 + \left(1+\Ocal(x^2)\right)dy^2 + \Ocal(z^2)\right].
\ee
But this metric looks just like pure AdS to quadratic order in~$x$ and~$z$ (which by construction are both small near the line~$\partial B$ which anchors the flow).  It is therefore quite reasonable to expect that just as for the half-plane flow~\eqref{eq:halfplaneflow}, boundary terms will not contribute to~$\dot{I}_H$.  This will precisely be the case, as we will show in the remainder of this Section.

\subsubsection{Monotonicity of Reduced Hawking Mass}
\label{subsec:monotoneboundaryflows}

Because~$I_H$ is finite when evaluated on any boundary-anchored surface (as shown in Appendix~\ref{app:massreg}),~$\dot{I}_H$ must be finite as well.  However, in order to control the boundary terms in~$\dot{I}_H$, it will prove useful to define a regulated IMC flow in which the asymptotic regions of the flow surfaces are cut off by some regulator.  We may then analyze the behavior of boundary terms in~$\dot{I}_H$ as the regulator is removed.

\begin{figure}[t]
\centering
\includegraphics[width=0.4\textwidth,page=6]{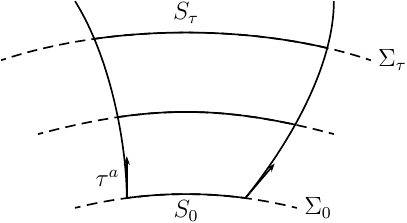}
\caption{The dashed lines represent an arbitrary IMC flow~$\Sigma_\tau$; the solid surfaces~$S_\tau$ are formed by selecting a portion~$S_0$ of an initial flow surface~$\Sigma_0$ and then transporting it along integral curves of a flow vector field~$\tau^a$.  This then defines an IMC flow~$S_\tau$ of surfaces with boundary.}
\label{fig:boundedflow}
\end{figure}

To that end, first consider an arbitrary IMC flow~$\Sigma_\tau$ (which need not be boundary-anchored).  We may define an IMC flow with boundary by first specifying some region~$S_0 \subset \Sigma_0$ of the initial flow surface, and then mapping~$S_0$ to the rest of the flow along the integral curves of some flow vector field~$\tau^a$ of our choosing.  The result is a family of flow surfaces~$S_\tau$ with boundary, as shown in Figure~\ref{fig:boundedflow}.  Of course, this procedure will only be well-defined if~$n_a \tau^a > 0$ everywhere, so that the integral curves of~$\tau^a$ never ``turn around''.  A natural condition to impose on~$\tau^a$ in order to guarantee this condition is
\be
\tau^a \grad_a \tau = 1,
\ee
which has the added benefit of allowing us to compute~$\tau$-derivatives along the flow as Lie derivatives along~$\tau^a$, as we did along~$\hat{n}^a$ for compact surfaces.  It follows from the definition of IMC flow that~$\tau^a$ can be decomposed into components normal and tangent to the~$S_\tau$ as
\be
\tau^a = \frac{n^a}{^S \! K} + \tau^a_\parallel,
\ee
where the unconstrained tangential component~$\tau^a_\parallel$ encodes how the boundary of the~$S_\tau$ moves ``in'' or ``out'' along the flow.

Using the Gauss-Bonnet theorem~\eqref{eq:GaussBonnet}, the reduced Hawking mass of the~$S_\tau$ can be written as
\be
\label{eq:redHawkingmass2}
I_H[S] = 8 \pi \chi_S + \int_S \left(\frac{4}{\ell^2} - \, ^S \! K^2\right) + 4 \int_{\partial S} k_g.
\ee
It is then straightforward to use the general time evolution formula~\eqref{eq:Fdot}, taking both~$f = 4/\ell^2 - \, ^S \! K^2$ and~$k_g$, to obtain the time evolution of~$I_H$ along the flow~$\tau^a$.  Using~\eqref{eq:GaussCodazzi},~\eqref{eq:GaussCodazzi2}, and the Gauss-Bonnet theorem, we obtain
\begin{subequations}
\begin{multline}
\label{subeq:firstterm}
\frac{d}{d\tau} \int_S \left(\frac{4}{\ell^2} - \, ^S K^2\right) = \int_{\partial S} \left( -2 v^a \grad_a \ln \, ^S \! K  + v_a \tau^a_\parallel \left(\frac{4}{\ell^2} - \, ^S \! K^2\right) \right) \\ -\frac{1}{2} I_H[S] + \int_S \left(\left| ^N \! K_{ab} \right|^2 + \left| ^S \widehat{K}_{ab} \right|^2 + 2 \left|D_a \ln \, ^S \! K \right|^2 + 16\pi G_N T_{ab} t^a t^b \right),
\end{multline}
\be
\frac{d}{d\tau} \int_{\partial S} k_g = \int_{\partial S} \left(\tau^a_\perp \grad_a k_g - k_g \tau^a_\perp a_a \right),
\ee
\end{subequations}
where~$v^a$ is the unit outward-pointing normal to~$\partial S$ in~$S$ and~$\tau^a_\perp \equiv n^a/\, ^S \! K + v^a(v_b \tau^b_\parallel)$ is the component of~$\tau^a$ normal to~$\partial S$.  Noting once again the positivity of the integral over~$S$ on the right-hand side of~\eqref{subeq:firstterm}, we can combine these results to obtain
\be
\label{eq:dotredHawkingmass2}
\dot{I}_H[S] + \frac{1}{2} I_H[S] \geq I_1 + I_2,
\ee
where we have split the boundary terms into two contributions:
\begin{subequations}
\label{eqs:I12}
\bea
I_1 &\equiv \int_{\partial S} \frac{1}{^S \! K} \left(-2 v^a \grad_a \, ^S \! K +  4n^a \grad_a k_g - 4 k_g n^a a_a\right), \\
I_2 &\equiv \int_{\partial S} v_a \tau^a_\parallel \left(\frac{4}{\ell^2} - \, ^S \! K^2 - 4k_g^2 + 4v^b \grad_b k_g \right).
\eea
\end{subequations}
The first term is independent of the details of the choice of flow field~$\tau^a$, while the second depends on a particular choice of~$\tau^a_\parallel$.

Monotonicity of the flow will therefore be preserved as long as~$I_1 + I_2 \geq 0$.  For a general flow, one could imagine choosing the boundary~$\partial(\cup S_\tau)$ of the flow appropriately to guarantee this property; for instance, requiring~$\tau^a_\parallel = 0$ immediately kills the boundary term~$I_2$.  In such a case, some partial results by Marquardt~\cite{Mar15} show that a quantity closely related to the Hawking mass is monotonic along the flow as long as~$\partial(\cup S_\tau)$ is convex.  However, as our primary purpose is to specialize these results to regulated boundary-anchored flows, we will not pursue this direction here.  Instead, we will now prove the main result of this Section: that~$I_{1,2}$ vanish on any regular boundary-anchored flow.  We emphasize that this result does not require that the flow begin on an extremal surface; it is therefore stronger than what might have been intuitively expected from the example of Section~\ref{subsec:bndryflow}.

\begin{thm}
\label{thm:monotonicity}
Consider a maximal-volume time slice~$N$ of a (3+1)-dimensional asymptotically locally AdS spacetime obeying Einstein's equation~\eqref{eq:Einstein} with a stress tensor that satisfies the weak energy condition.  Let~$\Sigma_\tau$ be a regular boundary-anchored IMC flow on~$N$ which is anchored to the AdS boundary on a smooth curve~$\partial B$.  Then the reduced Hawking mass of these surfaces obeys
\be
\label{eq:Idotbndry}
\dot{I}_H[\Sigma] + \frac{1}{2} I_H[\Sigma] \geq 0,
\ee
and therefore~$e^{\tau/2} I_H[\Sigma]$ is non-decreasing along the flow.
\end{thm}

\begin{proof}
By the stated assumptions, equation~\eqref{eq:dotredHawkingmass2} applies, and thus we need only evaluate the boundary terms~$I_1$ and~$I_2$.  Since the flow is singular at~$\partial B$, we begin by regulating it: in a neighborhood of~$\partial B$, we excise the asymptotic region of each flow surface~$\Sigma$ by introducing a one-parameter family of cutoff curves~$\partial \Sigma^\eps$ chosen so that~$\partial \Sigma^\eps \to \partial B$ as~$\eps \to 0$.  We may then introduce the unit vector field~$u^a$ tangent to~$\partial \Sigma^\eps$ and the outward-pointing unit normal vector field~$v^a$, as shown in Figure~\ref{fig:cutoffs}.  By construction,~$u^a$,~$v^a$, and~$n^a$ form an orthonormal basis of~$N$ in a neighborhood of~$\partial B$.

\begin{figure}[t]
\centering
\includegraphics[width=0.7\textwidth,page=7]{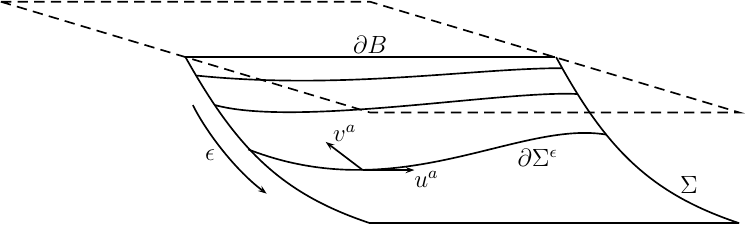}
\caption{The regulated IMC flow introduced in the proof of Theorem~\ref{thm:monotonicity}.  Here~$\Sigma$ is an IMC flow surface anchored to the AdS boundary (shown in dashed lines) on the curve~$\partial B$.  To regulate the region of~$\Sigma$ near~$\partial B$, we introduce a one-parameter family of curves~$\partial \Sigma^\eps$ which approach~$\partial B$ as~$\eps \to 0$.  The unit tangent and normal vector fields~$u^a$ and~$v^a$ to the~$\partial \Sigma^\eps$ define an orthonormal basis on~$\Sigma$ in a neighborhood of~$\partial B$.}
\label{fig:cutoffs}
\end{figure}

We now introduce a coordinate system on~$N$ adapted to~$\partial B$.  Consider any conformal frame in which~$\partial B$ has finite length; because~$N$ is asymptotically locally hyperbolic (since it is extremal, and thus by Lemma~\ref{lem:extremalorthog} it intersects the boundary orthogonally), near~$\partial N$ we may introduce a Fefferman-Graham coordinate~$z$ associated to this frame.  We then take the boundary coordinate~$y$ to be the proper distance along~$\partial B$ in this frame, and the boundary coordinate~$x$ to be a Gaussian normal coordinate off of~$\partial B$.  The metric on~$N$ near~$\partial B$ then takes the form
\be
ds^2_N = \frac{\ell^2}{z^2} \left[dz^2 + dx^2 + \left(1 - 2 x \, a(y) + \Ocal(x^2)\right)dy^2 + \Ocal(z^2)\right],
\ee
where~$a(y)$ is the geodesic acceleration of~$\partial B$ in this conformal frame.  Since~$\partial B$ is smooth and has finite length,~$a(y)$ is everywhere finite and~$y$ has finite range.

As in the case of the pure AdS flow~\eqref{eq:halfplaneflow}, it is convenient to introduce polar-like coordinates~$(\xi,\theta)$ via~$z = \xi \sin\theta$,~$x = \xi \cos\theta$ so that~$\partial B$ lies at~$\xi = 0$.  Then the metric becomes
\be
ds^2_N = \frac{\ell^2}{\xi^2 \sin^2\theta}\left[d\xi^2 + \xi^2 d\theta^2 + \left(1 - 2 \xi\cos\theta \, a(y)\right)dy^2 + \cdots \right],
\ee
where~$\cdots$ are terms subleading in~$\xi$.  In these coordinates, it is straightforward to study the behavior of a general IMC flow near~$\partial B$.  First, recall that by definition~\ref{def:boundaryIMC}, the flow time~$\tau$ can be written as
\be
\tau(\xi,\theta,y) = f(\theta) + \tau_\mathrm{reg}(\xi,\theta,y).
\ee
The requirement that~$\tau_\mathrm{reg}$ be single-valued at~$\xi = 0$ implies that~$\tau_\mathrm{reg}|_{\xi = 0} = c_0(y)$ for some~$c_0(y)$; the requirement that~$|\gradt \tau_\mathrm{reg}|$ be finite with respect to the conformally compactified metric~$\tilde{h}_{ab} \equiv (z/\ell)^2 h_{ab}$ is consistent with single-valuedness and additionally implies that~$\partial_\xi \tau_\mathrm{reg}$ and~$\partial_y \tau_\mathrm{reg}$ be finite at~$\xi = 0$.  These requirements thus imply the expansion
\be
\label{eq:flowsurface}
\tau(\xi,\theta,y) = f(\theta) + c_0(y) + \xi \tau_1(\theta,y) + o(\xi),
\ee
where we use little-o notation, defined here as~$\lim_{\xi \to 0} \xi^{-p} o(\xi^p) = 0$.

Then the flow equation~\eqref{eq:flowpde}, which can be thought of as an evolution equation in~$\theta$, allows us to solve for~$f(\theta)$ and~$\tau_1$.  We obtain
\be
\label{eq:taui}
f(\theta) = \ln\csc^2\theta, \qquad \tau_1(\theta,y) = c_1(y) - a(y) \cos\theta,
\ee
where like~$c_0(y)$,~$c_1(y)$ is an unconstrained constant of integration that can be specified by providing initial data for the flow (the constant of integration from~$f(\theta)$ can be absorbed into~$c_0(y)$, so we omit it).  Note that unsurprisingly, the leading-order term~$f(\theta)$ reproduces the pure AdS flow~\eqref{eq:halfplaneflow} found above.

With the solutions~\eqref{eq:taui}, we may now compute the asymptotic behavior of all the geometric objects necessary to evaluate the boundary terms~$I_1$ and~$I_2$.  The unit normal vector to the flow is given by~$n_a = \grad_a \tau/|\grad \tau|$, which in these coordinates we find has components
\begin{subequations}
\be
\begin{pmatrix} n^\xi \\ n^\theta \\ n^y \end{pmatrix} = \frac{\sin\theta}{2\ell} \begin{pmatrix} \tan\theta (c_1(y) - a(y)\cos\theta)\xi^2 + o(\xi^2) \\ -2 + o(\xi) \\ c_0'(y) \tan\theta \, \xi^2 + o(\xi^2) \end{pmatrix}.
\ee
To obtain the vector fields~$u^a$ and~$v^a$, recall that the cutoff curves~$\partial \Sigma^\eps$ were required to approach~$\partial B$ as~$\eps \to 0$.  More precisely, this implies that~$\tilde{h}_{ab} \tilde{v}^a (\partial_y)^b|_{\partial B} = 0$, where~$\tilde{v}^a \equiv \ell v^a/z$ is the unit (with respect to~$\tilde{h}_{ab}$) outward-pointing normal vector to~$\partial B$ on~$\Sigma$ in the conformally compactified space (which is regular at~$\partial B$).  In terms of the components of~$v^a$ in the uncompactified space, this condition requires~$v^y = o(\xi)$.  Combined with the conditions that~$v^a$ and~$u^a$ be normalized and orthogonal to each other and to~$n^a$, we then obtain
\bea
\begin{pmatrix} v^\xi \\ v^\theta \\ v^y \end{pmatrix} &= \frac{\sin\theta}{2\ell} \begin{pmatrix} -2\xi + o(\xi^2) \\ (a(y) \sin\theta - c_1(y) \tan\theta)\xi + o(\xi) \\ \alpha(\theta,y) \xi^2 + o(\xi^2) \end{pmatrix}, \\
\begin{pmatrix} u^\xi \\ u^\theta \\ u^y \end{pmatrix} &= \frac{\sin\theta}{2\ell} \begin{pmatrix} \alpha(\theta,y) \xi^2 + o(\xi^2) \\ c_0'(y) \tan\theta \, \xi + o(\xi) \\ 2 \xi + 2a(y) \cos\theta \, \xi^2 + o(\xi^2) \end{pmatrix},
\eea
\end{subequations}
where~$\alpha(\theta,y)$ is an arbitrary function which encodes the freedom in choosing the family of cutoff surfaces~$\partial \Sigma^\eps$.

The mean curvature of~$\Sigma$ and the geodesic curvature of~$\partial \Sigma^\eps$ are then
\begin{subequations}
\bea
^\Sigma \! K &\equiv \grad_a n^a = \frac{1}{\ell}\left[2\cos\theta - a(y) \sin^2\theta \, \xi + o(\xi)\right], \\
k_g &\equiv v_a u^b \grad_b u^a = -\frac{\sin\theta}{2\ell}\left[2 + (c_1(y) + a(y) \cos\theta)\xi + o(\xi)\right].
\eea
\end{subequations}
We may now easily evaluate the integrands of the boundary terms~\eqref{eqs:I12}; we obtain
\begin{subequations}
\bea
\frac{1}{^\Sigma \! K} \left(-2 v^a \grad_a \, ^\Sigma \! K +  4n^a \grad_a k_g - 4 k_g n^a a_a\right) = o(\xi), \label{subeq:I1integrand} \\
\frac{4}{\ell^2} - \, ^\Sigma \! K^2 - 4k_g^2 + 4v^b \grad_b k_g = o(\xi). \label{subeq:I2integrand}
\eea
\end{subequations}
Now, the measure in~$I_1$ and~$I_2$ is the proper length element~$ds$ along~$\partial \Sigma^\eps$, which is related to the coordinate length element~$dy$ as
\be
ds = \left(\frac{1}{\xi \sin\theta} + \Ocal(1)\right) dy.
\ee
Recalling that~$y$ has finite range, combining with~\eqref{subeq:I1integrand} immediately yields
\be
I_1 = \int_{\partial \Sigma^\eps} o(\xi^0) \, dy \to 0 \mbox{ as } \eps \to 0.
\ee
(Recall that~$\xi \to 0$ as~$\eps \to 0$.)  Next, we require that~$\tau^a_\parallel$ (which may be chosen arbitrarily) be finite in the limit~$\xi \to 0$.  This implies that~$v_a \tau^a_\parallel = \Ocal(\xi^0)$, and so using~\eqref{subeq:I2integrand} we also obtain
\be
I_2 = \int_{\partial \Sigma^\eps} o(\xi^0) \, dy \to 0 \mbox{ as } \eps \to 0.
\ee
Thus~$I_1$ and~$I_2$ vanish on any regular boundary-anchored flow; comparison with~\eqref{eq:dotredHawkingmass2} completes the proof.
\end{proof}

Having shown that boundary terms don't contribute to the evolution of a regular boundary-anchored flow, it is not much work to show that in fact the only regular boundary-anchored IMC flow that saturates the inequality~\eqref{eq:Idotbndry} is the flow~\eqref{eq:halfplaneflow} in pure AdS:

\begin{lem}
\label{lem:saturate}
The inequality~\eqref{eq:Idotbndry} is saturated along a regular boundary-anchored IMC flow if and only if~$N$ is (a portion of) a static slice of pure AdS and the~$\Sigma_\tau$ are homogeneous hyperbolic slices.
\end{lem}

\begin{proof}
Because boundary terms do not contribute to~$\dot{I}_H$, we have from~\eqref{eq:Idot} that the inequality will be saturated if and only if~$^N \! K_{ab}$,~$^\Sigma \! \widehat{K}_{ab}$, and~$T_{ab} t^a t^b$ vanish and~$^\Sigma \! K$ is constant on each~$\Sigma_\tau$.  It is convenient to work with an ADM-like decomposition of the metric on~$N$ adapted to the flow:
\be
\label{eq:saturateslice}
ds^2_N = \frac{d\tau^2}{^\Sigma \! K^2} + \sigma_{ij}(\tau,x) dx^i \, dx^j,
\ee
where~$\tau$ is the flow time,~$x^i$ are coordinates on the~$\Sigma_\tau$, and~$\sigma_{ij}$ are the components of the induced metric on the~$\Sigma_\tau$.  Note that the~$g_{\tau\tau}$ metric component above was fixed by the flow equation~\eqref{eq:flow}.

We then find
\be
^\Sigma \! \widehat{K}_{ij} = \frac{^\Sigma \! K}{2} \left(\sigma_{ij}' - \sigma_{ij}\right),
\ee
where a prime denotes a~$\tau$ derivative.  Thus the requirement that~$^\Sigma \! \widehat{K}_{ij}$ vanish implies~$\sigma_{ij}(\tau,x) = e^{\tau} \, \bar{\sigma}_{ij}(x)$ for some~$\tau$-independent metric~$\bar{\sigma}_{ij}$.  Next, from~\eqref{eq:GaussCodazzi} the vanishing of~$T_{ab} t^a t^b$ and~$^N \! K_{ab}$ imply~$^N \! R =  -6/\ell^2$, which using~\eqref{eq:saturateslice} yields
\be
\label{eq:Kdiffeq}
\left( ^\Sigma \! K^2 \right)' + \frac{3}{2} \, ^\Sigma \! K^2 - e^{-\tau} \, ^\Sigma \! \bar{R}(x) = \frac{6}{\ell^2},
\ee
where~$^\Sigma \! \bar{R}$ is the Ricci scalar of~$\bar{\sigma}_{ij}$.  Since~$^\Sigma \! K$ is independent of~$x$, the above implies that~$^\Sigma \! \bar{R}$ is in fact constant, which we rescale (via a shift in~$\tau$) to~$^\Sigma \! \bar{R} = 2k/\ell^2$ for~$k = 0$ or~$\pm 1$.  Thus the~$\Sigma_\tau$ are (locally) homogeneous spaces.  Solving~\eqref{eq:Kdiffeq} then yields
\be
^\Sigma \! K^2 = \frac{4}{\ell^2}\left(1 + k e^{-\tau} - \frac{2m}{\ell} e^{-3\tau/2}\right),
\ee
where~$m$ is a constant of integration.  Inserting this expression back into~\eqref{eq:saturateslice} and converting to a new coordinate~$r = \ell e^{\tau/2}$, we finally obtain
\be
ds^2_N = \frac{dr^2}{r^2/\ell^2 + k - 2m/r} + r^2 dH^2_k,
\ee
where~$dH^2_k$ is the metric on a homogeneous two-dimensional space with Ricci scalar~$2k$.  We immediately recognize the above as the metric on a static slice of Schwarzschild-AdS.

\begin{figure}[t]
\centering
\subfigure[]{
\includegraphics[width=0.25\textwidth,page=8]{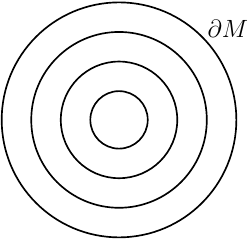}
\label{subfig:sphere}
}
\hspace{0.5cm}
\subfigure[]{
\includegraphics[width=0.25\textwidth,page=9]{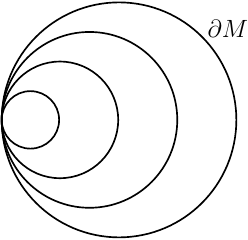}
\label{subfig:plane}
}
\hspace{0.5cm}
\subfigure[]{
\includegraphics[width=0.25\textwidth,page=10]{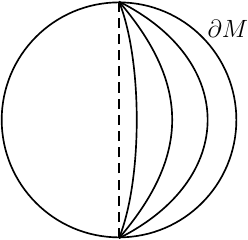}
\label{subfig:hyperplane}
}
\caption{A static time slice of pure AdS and homogeneous IMC flows on it.  For~$k = 1$ (left), the flow surface are spheres, which are not boundary anchored.  For~$k = 0$ (middle), the flow surface are planar slices of Poincar\'e AdS, which reach the boundary only at a point; they do not constitute a regular boundary-anchored flow.  For~$k = -1$ (right), the flow surfaces are hyperbolic slices of a Rindler patch of AdS, which are anchored to the boundary on the same curve as the Rindler horizon (dashed).}
\label{fig:AdSIMCs}
\end{figure}

To check which of these flows is boundary-anchored on a curve~$\partial B$, we may immediately exclude the case~$k = 1$ of spherical flow surfaces, since spheres have finite area and therefore cannot be boundary-anchored (Figure~\ref{subfig:sphere}).  Next, recall from Appendix~\ref{app:massreg} that~$I_H[\Sigma]$ must be finite on any boundary-anchored surface.  In the cases~$k = 0$,~$-1$, the space~$dH^2_k$ is not compact\footnote{One could consider compact quotients, but the corresponding flows cannot be boundary-anchored.}, and thus~$I_H[\Sigma]$ will only be finite if its integrand (which is constant on each flow surface) vanishes.  For the flows above, we find that
\be
2 \, ^\Sigma \! R + \frac{4}{\ell^2} - \, ^\Sigma \! K^2 = \frac{8m}{\ell^3} e^{-3\tau/2},
\ee
and thus a necessary condition for the flow to be boundary-anchored is~$m = 0$; i.e.~in these coordinates,~$N$ must be a Poincar\'e ($k = 0$) or Rindler ($k = -1$) slice of pure AdS.  But flat spatial slices of the Poincar\'e patch reach the AdS boundary only at a point (Figure~\ref{subfig:plane}), violating the requirement that~$\partial B$ be a curve; thus~$k = 0$ is excluded as well.  The only remaining case is~$k = -1$, where the flow surfaces are anchored at the boundary on the same curve as (the bifurcation surface of) the corresponding Rindler horizon (Figure~\ref{subfig:hyperplane}); in a different coordinate system, this is exactly the flow~\eqref{eq:halfplaneflow} constructed above, which is clearly a regular boundary-anchored flow.  (Note that~$\partial B$ has infinite extent in the natural conformal frame associated with the coordinates of~\eqref{eq:PoincareAdS}, but this is easily remedied by switching to a different frame which is regular everywhere on the line~$\xi = 0$.)
\end{proof}

\section{A Bound on Energy Density and Entanglement Entropy}
\label{sec:bounds}

We may now exploit the monotonicity of the reduced Hawking mass under boundary-anchored IMC flows to obtain our bound.  As mentioned in Section~\ref{sec:intro}, the idea is to consider the IMC flow along a maximal-volume slice of an entanglement wedge.  By starting the flow on the HRT surface and letting it flow to the boundary, the monotonicity property of Theorem~\ref{thm:monotonicity} yields an inequality between the renormalized area of the HRT surface and of a weighted integral of the local energy density on the boundary.  In Section~\eqref{subsec:renarea} we showed that the reduced Hawking mass on any boundary-anchored extremal surface is
\be
\label{eq:initialmH}
I_H[\Sigma] = 4\left(2\pi \chi_\Sigma + \frac{\Acal[\Sigma]}{\ell^2}\right),
\ee
where~$\chi_\Sigma$ is the Euler characteristic of~$\Sigma$ and~$\Acal[\Sigma]$ is the renormalized area of~$\Sigma$.  We now compute the asymptotic behavior of the reduced Hawking mass, and thereby obtain the promised bound.

\subsection{Asymptotic Hawking Mass}
\label{subsec:finalmH}

\begin{figure}[t]
\centering
\includegraphics[width=0.6\textwidth,page=11]{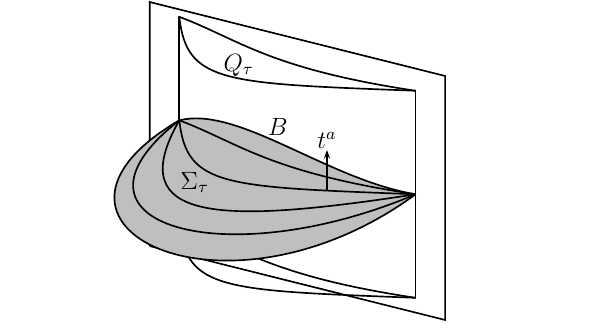}
\caption{The surfaces used in the statement and proof of Lemma~\ref{lem:finalmH}.}
\label{fig:finalmH}
\end{figure}

\begin{lem}
\label{lem:finalmH}
Let~$(M,g_{ab})$ be a (3+1)-dimensional asymptotically locally AdS spacetime which satisfies Einstein's equation~\eqref{eq:Einstein} with a stress tensor~$T_{ab}$ whose components in any orthonormal frame decay near~$\partial M$ as~$T_{\hat{\mu} \hat{\nu}} = o(\Omega^3)$.  Then assume the following objects, depicted in Figure~\ref{fig:finalmH}, exist:
\begin{compactitem}
	\item A two-dimensional surface~$B \subset \partial M$ on a moment of time symmetry of the boundary (that is,~$B$ lies on a moment of time symmetry of some representative of the conformal class of the boundary);
	\item A time slice~$N$ of the bulk which intersects~$B$ and is extremal in a neighborhood of~$B$; and
	\item A regular boundary-anchored IMC flow~$\Sigma_\tau$ on~$N$ anchored on~$\partial B$ and asymptoting to~$B$.
\end{compactitem}
Then the asymptotic behavior of the reduced Hawking mass of the flow surfaces obeys
\be
\label{eq:finalmH}
\lim_{\tau \to \infty} e^{\tau/2} I_H[\Sigma_\tau] = \frac{32\pi G_N}{\ell^2} \int_B \omega \left\langle \varepsilon \right\rangle,
\ee
where the integral is evaluated on~$B$ in any conformal frame with associated defining function~$\Omega$;~$\langle \varepsilon \rangle$ is (the expectation value of) the local energy density on~$B$ in this frame as measured by observers moving orthogonally to~$B$; and
\be
\omega \equiv \lim_{\tau \to \infty} \ell \, e^{\tau/2} \, \Omega|_{\Sigma_\tau}.
\ee
Moreover, as long as the frame defined by~$\Omega$ is regular everywhere on~$B \cup \partial B$,~$\omega$ is finite and positive in~$B$, and vanishes on~$\partial B$.
\end{lem}

\begin{proof}
The derivation is a generalization of that in~\cite{FisHic16}.  First, introduce a foliation of the near-boundary region with timelike hypersurfaces~$Q_\tau$, as shown in Figure~\ref{fig:finalmH}.  We take the~$Q_\tau$ to be normal to~$N$ and to intersect~$N$ on the IMC flow surfaces~$\Sigma_\tau$.  On these~$Q_\tau$ we may construct the Balasubramanian-Kraus stress tensor~\cite{BalKra99} (see also~\cite{deHSol00})\footnote{Because the bulk stress tensor falls off like~$o(\Omega^3)$,~$\T_{ab}$ gets no contribution from bulk matter counterterms.}
\be
\label{eq:BalaKraus}
\T_{ab} = \frac{1}{8\pi G_N}\left(- \, ^Q \! K_{ab} + \, ^Q \! K q_{ab} - \frac{2}{\ell} q_{ab} + \ell \, ^Q \! G_{ab}\right),
\ee
where~$q_{ab}$ is the induced metric on the~$Q_\tau$ and~$^Q \! K_{ab}$ and~$^Q \! G_{ab}$ are their extrinsic curvature and Einstein tensor, respectively.  Note that as defined,~$\T_{ab}$ is a bulk object, but it asymptotically approaches the boundary stress tensor (up to a conformal factor) as the~$Q_\tau$ approach the boundary.  Specifically, in a conformal frame specified by a defining function~$\Omega$, we have\footnote{This expression follows from the fact that by construction, the density~$\sqrt{\sigma} \T_{ab} t^a \xi^b$ is a conformal invariant for any~$\xi^a$ left unchanged by conformal transformations, so that when~$\xi^a$ is a conformal Killing vector field the integral of~$\sqrt{\sigma} \T_{ab} t^a \xi^b$ over a boundary Cauchy slice generates a conserved charge~\cite{BalKra99}.}
\be
\label{eq:Tbndryfalloff}
\lim_{\Omega \to 0} \Omega^{-3} \T_{ab} t^a t^b = \langle T^\mathrm{bndy}_{ab} \rangle \tilde{t}^a \tilde{t}^b \equiv \left\langle \varepsilon \right\rangle,
\ee
where~$\langle T^\mathrm{bndy}_{ab} \rangle$ is the expectation value of the boundary stress tensor in this frame and~$\langle \varepsilon \rangle$ is the corresponding energy density as measured by observers with velocity~$\tilde{t}^a$ normal to~$B$.  We pause here to note the additional falloffs
\be
\label{eq:GKfalloffs}
^Q \! G_{ab} t^a t^b = \Ocal(\Omega^2), \qquad \, ^N \! K_{\hat{\mu}\hat{\nu}} = \Ocal(\Omega^2).
\ee
The first follows straightforwardly from the conformal transformation properties of the Einstein tensor, while the second (for the components of~$^N \! K_{ab}$ in any orthonormal frame) follows from Lemma~\ref{lem:Kfalloff} in Appendix~\ref{app:Kconformal} and relies on the fact that~$B$ lies on a moment of time symmetry.

Now, the extrinsic curvature~$^\Sigma \! K_{ab}$ of the~$\Sigma_\tau$ in~$N$ is equal to the projection of the extrinsic curvature~$^Q \! K_{ab}$ onto~$N$; that is,
\be
^\Sigma \! K_{ab} = {\sigma_a}^c {\sigma_b}^d \, ^Q \! K_{cd},
\ee
where recall that~$\sigma_{ab} = q_{ab} + t_a t_b$ is the induced metric on the~$\Sigma_\tau$.  Thus the mean curvature of the~$\Sigma_\tau$ is related to that of the~$Q_\tau$ by
\be
^\Sigma \! K = \, ^Q \! K + \, ^Q \! K_{ab} t^a t^b.
\ee
This result then allows us to use~\eqref{eq:BalaKraus} to write the mean curvature of the~$\Sigma_\tau$ in terms of~$\T_{ab}$:
\be
\label{eq:KsigmaTab}
^\Sigma \! K = \frac{2}{\ell} - 8\pi G_N \T_{ab} t^a t^b + \ell \, ^Q \! G_{ab} t^a t^b.
\ee
Likewise, we may use the Gauss-Codazzi equations to express the Ricci scalar~$^\Sigma \! R$ on the~$\Sigma_\tau$ in terms of the geometry of the~$Q_\tau$; writing the result in terms of the extrinsic curvature of~$N$ and exploiting the fact that~$^N \! K = 0$, we obtain
\be
\label{eq:sigmaR}
^\Sigma \! R = 2 \, ^Q \! G_{ab} t^a t^b + \left| ^N \! K_{ab} \right|^2 - 2\left| ^N \! K_{ab} n^b \right|^2.
\ee
Inserting~\eqref{eq:KsigmaTab} and~\eqref{eq:sigmaR} into the definition~\eqref{eq:redHawkingmass} of the reduced Hawking mass and then using~\eqref{eq:GKfalloffs}, we obtain for arbitrary defining function~$\Omega$
\be
\label{eq:masymp}
I_H[\Sigma_\tau] = \frac{32\pi G_N}{\ell} \int_{\Sigma_\tau} \left(\T_{ab} t^a t^b + \Ocal(\Omega^4)\right).
\ee
Multiplying by~$e^{\tau/2}$, taking the limit~$\tau \to \infty$ (so~$\Omega \to 0$), and using~\eqref{eq:Tbndryfalloff}, we recover~\eqref{eq:finalmH}.

Finally, that~$\omega$ vanishes on~$\partial B$ follows immediately from the fact that at any finite~$\tau$,~$e^{\tau/2} \, \Omega|_{\partial B} = 0$ (assuming~$\Omega$, and therefore the conformal frame, is regular at~$\partial B$).  It is clearly positive in~$B$, and that it is finite follows from the asymptotics of the IMC flow studied in e.g.~\cite{LeeNev15,FisHic16}.
\end{proof}

Note the importance of~$B$ lying on a moment of time symmetry (though we emphasize that~$B$ only needs to lie on a moment of time symmetry of the boundary metric, not of the entire state).  This feature was necessary to guarantee that the terms in~\eqref{eq:sigmaR} involving~$^N \! K_{ab}$ vanish sufficiently rapidly near the boundary, as per Lemma~\ref{lem:Kfalloff}.  Without this condition, these terms would fall off too slowly and~$e^{\tau/2} I_H[\Sigma_\tau]$ would diverge.

\subsection{A Bound}
\label{subsec:bound}

\begin{thm}
\label{thm:bound}
Consider a (2+1)-dimensional holographic CFT whose bulk dual obeys Einstein's equation~\eqref{eq:Einstein} with a bulk stress tensor obeying the weak energy condition and whose components in any orthonormal frame fall off like~$T_{\hat{\mu}\hat{\nu}} = o(\Omega^3)$ near the asymptotic boundary.  Consider also a spatial region~$B$ on a moment of time symmetry of the CFT spacetime.  Assume the following objects exist:
\begin{compactitem}
	\item An extremal surface~$\Xcal_B$ with the same topology as~$B$ and anchored to~$\partial B$;
	\item A maximal-volume time slice~$N$ of the bulk with~$\partial N = \Xcal_B \cup B$; and
	\item A regular boundary-anchored IMC flow~$\Sigma_\tau$ on~$N$ with~$\Sigma_0 = \Xcal_B$ and~$\Sigma_{\tau \to \infty} = B$.
\end{compactitem}
Then
\be
\label{eq:bound}
S^\mathrm{ren}_B \leq \left\langle \Ecal_B \right\rangle - 8\pi^2 c_\mathrm{eff}\chi_B, \qquad \Ecal_B \equiv 2\pi \int_B \omega \varepsilon,
\ee
where~$\omega$ and~$\varepsilon$ are as defined in Lemma~\ref{lem:finalmH}, and~$\chi_B$ and~$S^\mathrm{ren}_B$ are the Euler characteristic and the renormalized entropy of~$B$, respectively.
\end{thm}

\begin{proof}
Because~$N$ is an extremal surface,~$T_{ab}$ obeys the weak energy condition, and~$\Sigma_\tau$ is a regular boundary-anchored IMC flow, we may apply Theorem~\ref{thm:monotonicity} to conclude that
\be
\lim_{\tau \to \infty} e^{\tau/2} I_H[\Sigma_\tau] \geq I_H[\Sigma_0].
\ee
Combining~\eqref{eq:initialmH} and Lemma~\ref{lem:finalmH}, we thus obtain
\be
\label{eq:tempinequality}
\frac{32 \pi G_N}{\ell^2} \int_B \omega \langle\varepsilon\rangle \geq 4\left(2\pi \chi_{\Xcal_B} + \frac{\Acal[\Xcal_B]}{\ell^2}\right),
\ee
with~$\Acal[\Xcal_B]$ the renormalized area of~$\Xcal_B$.  But by construction,~$\Xcal_B$ has the same topology as~$B$, so~$\chi_{\Xcal_B} = \chi_B$.  The HRT formula also requires that the HRT surface~$X_B$ whose area gives the entanglement entropy of~$B$ have less area than any other extremal surface anchored to~$\partial B$, and thus~$A[\Xcal_B] \geq A[X_B]$.  But since both~$\Xcal_B$ and~$X_B$ are anchored at~$\partial B$, the UV-divergent pieces in their areas match, and thus the renormalized areas must also obey~$\Acal[\Xcal_B] \geq \Acal[X_B]$.  Writing~$S^\mathrm{ren}_B = \Acal[X_B]/4G_N$,~\eqref{eq:tempinequality} yields~\eqref{eq:bound}.
\end{proof}

Before moving on to interpretations of this result, let us briefly comment on the assumptions of the proof.  First, note that the HRT surface~$X_B$ need not always have the same topology as~$B$ (for instance, if~$B$ consists of a ``fat'' annulus), hence our use of~$\Xcal_B$ instead.  However, like~$X_B$,~$\Xcal_B$ will exist quite generically, as will the extremal time slice~$N$.  Indeed, in the particular case where~$\Xcal_B$ is the HRT surface, the homology constraint guarantees the existence of a spatial slice~$N_\mathrm{HRT}$ with boundary~$\Xcal_B \cup B$; by varying this slice in time, its area can be maximized, yielding the extremal surface~$N$.  However, we have left these assumptions explicit because there are known to be fine-tuned spacetimes exhibiting strong time dependence in the bulk where the HRT surface and volume-minimizing slices do not exist~\cite{FisMar14}.

Next, we note that for a classical bulk, the weak energy condition is a reasonable requirement on the bulk stress tensor~$T_{ab}$ (though not as desirable as, say, the null energy condition).  More relevantly, the falloff of~$T_{ab}$ near the boundary imposes constraints on any CFT operators dual to bulk matter fields.  As a brief example, consider the asymptotic behavior of a free bulk scalar field of mass~$m^2 = \Delta(\Delta-3)/\ell^2$ in some Fefferman-Graham coordinate~$z$ (see e.g.~\cite{FisKel12}):
\be
\label{eq:phiexpansion}
\phi = z^{3-\Delta}\left(\phi^{(0)} + z^2 \phi^{(2)} + \cdots\right) + \begin{cases} z^\Delta \ln z \left(\psi^{(0)} + z^2 \psi^{(2)} + \cdots\right) \mbox{ if } \Delta = \frac{2n+1}{2}, \\ z^\Delta \left(\phi^{(2\Delta -3)} + z^2 \phi^{(2\Delta -1)} + \cdots\right) \mbox{ otherwise},
\end{cases}
\ee
where~$\Delta \geq 3/2$ and~$n$ is a positive integer.  Consider the cases~$\Delta > 3/2$ and~$\Delta = 3/2$ separately.  For~$\Delta > 3/2$, standard quantization in the AdS/CFT dictionary dictates that~$\phi^{(0)}$ and~$\phi^{(2\Delta -3)}$  are the source and one-point function of a CFT scalar operator of conformal weight~$\Delta$, and determine all other terms in the above expansions.  Then very schematically, the
leading-order behavior of the orthonormal components of the bulk stress tensor is
\be
T_{\hat{\mu}\hat{\nu}} \sim z^{2(3-\Delta)} \left(\phi^{(0)}\right)^2 + \cdots,
\ee
and thus the falloff condition on~$T_{ab}$ requires~$\phi^{(0)} = 0$, and hence a vanishing source.  
In the case~$\Delta = 3/2$, the CFT source and one-point function are~$\phi^{(0)}$ and~$\psi^{(0)}$, which determine all other terms in~\eqref{eq:phiexpansion}; the falloff condition on the stress tensor imposes that both~$\phi^{(0)}$ and~$\psi^{(0)}$ must vanish, implying $\phi = 0$ in the bulk. 
Thus the falloff condition on~$T_{ab}$ requires that no CFT sources be turned on for scalar operators dual to free fields with conformal dimension~$\Delta > 3/2$, and that neither source nor expectation value be turned on for scalar operators of dimension~$\Delta = 3/2$\footnote{In the range~$3/2 < \Delta < 5/2$, one can choose an alternative quantization in which~$\phi^{(0)}$ and~$\phi^{(2\Delta -3)}$ swap roles for a scalar operator of dimension~$\overline{\Delta} = 3-\Delta$~\cite{KleWit99,Wit01b,HolWar13}.  In CFT parlance, one would then say that our bound applies to any \textit{state} in which no scalar operator of conformal dimension~$1/2 < \overline{\Delta} < 3/2$ has a nonzero one-point function (though a source may be turned on).  However, this is somewhat unnatural, as one would prefer to know whether or not the bound will hold in a particular \textit{theory}, which is fixed by sources, not one-point functions.}.  The essence of this argument is unchanged for other types of matter, so we expect these statements hold for more general CFT sources and operators as well.

Finally, the strongest assumption required to obtain the above bound is the existence of a smooth IMC flow from~$\Xcal_B$ to~$B$.  As mentioned in Section~\ref{sec:IMCF}, it is known that IMC may become singular, and thus one needs to instead consider weak flows.  When the flow surfaces are compact, these weak flows are also known to preserve monotonicity of the Hawking mass~\cite{HuiIlm01,LeeNev15}, yielding Penrose inequalities.  It is likely that these constructions extend to the flows considered here: some recent work~\cite{Mar15} has shown that in Euclidean space, unique solutions to the weak IMC flow can be found in which the flow surfaces are constrained to be normal to some boundary, and moreover~\cite{Mar15} suggests that a Hawking mass-like functional is monotonic along these weak flows.  We therefore find it likely that there exists a weak formulation of our boundary-anchored flows which preserves monotonicity of~$e^{\tau/2}I_H[\Sigma]$.

\section{Beyond the Classical Limit}
\label{sec:perturbative}

The derivation of the bound~\eqref{eq:bound} was purely classical; that is, it assumed the existence of a classical bulk spacetime obeying Einstein's equation sourced by classical matter.  But one crucial feature of strict inequalities is that they cannot be violated by (sufficiently small) perturbative corrections.  In the present context, this corresponds to an intuitive expectation that when~\eqref{eq:bound} is obeyed as a \textit{strict} inequality, perturbative corrections cannot violate it.  The situation is more subtle when the bound is saturated, however; thus our purpose is now to investigate this intuition more carefully and argue that even as a weak inequality,~\eqref{eq:bound} holds for any states which are finite-order perturbations of a classical state.  This result relies crucially on the positivity of relative entropy, which we now quickly review.

\subsection{Positivity of Relative Entropy and the First Law}
\label{subsec:relativeent}

Consider some CFT region~$B$ and any two states~$\rho_B$ and~$\sigma_B$ on~$B$.  One can then define the relative entropy of~$\rho_B$ and~$\sigma_B$ as
\be
\label{eq:relativeent}
S(\sigma_B | \rho_B) \equiv \Tr(\sigma_B \ln \sigma_B) - \Tr(\sigma_B \ln \rho_B).
\ee
The relative entropy~$S(\sigma_B | \rho_B)$ essentially measures the distinguishability of~$\rho_B$ and~$\sigma_B$, and it therefore has the property that it is positive whenever~$\rho_B$ and~$\sigma_B$ are distinct\footnote{Note that as is customary, here and below we will neglect any potential subtleties associated with~$\rho_B$ and~$\sigma_B$ living in uncountably infinite-dimensional Hilbert spaces.} (see e.g.~\cite{Weh78} for a review).  Defining the modular Hamiltonian of~$\rho_B$ as~$H_B = -\ln \rho_B$, this positivity property can be rewritten as
\be
\label{eq:nonlinfirstlaw}
\Delta \left\langle H_B \right\rangle \geq \Delta S_B,
\ee
where\footnote{We remind the reader that here~$H_B$ always refers to the modular Hamiltonian of~$\rho_B$, and never to that of~$\sigma_B$.}
\be
\Delta \left\langle H_B \right\rangle \equiv \Tr(\sigma_B H_B) - \Tr(\rho_B H_B), \qquad \Delta S_B \equiv S_B(\sigma_B) - S_B(\rho_B).
\ee
Note that because the region~$B$ is unchanged and the area-law divergence in~$S_B$ is state-independent, we have that~$\Delta S_B = \Delta S^\mathrm{ren}_B$.

As a special case, note that since the relative entropy is positive and vanishes when~$\sigma_B = \rho_B$, it must be stationary under perturbations~$\sigma_B = \rho_B + \delta \rho$ for some small~$\delta \rho$.  Then the inequality~\eqref{eq:nonlinfirstlaw} is saturated to first order in~$\delta \rho$:
\be
\label{eq:firstlaw}
\delta \left\langle H_B \right\rangle = \delta S_B + \Ocal(\delta^2).
\ee
This so-called first law of entanglement, derived in~\cite{BlaCas13}, can be interpreted as a version of the first law of thermodynamics~$dE = T dS$ for nearby equilibrium states.

\subsection{Saturating the Bound}
\label{subsec:pureAdS}

Interestingly, in a certain special case we can derive the first law~\eqref{eq:firstlaw}.  Doing so requires first saturating our bound~\eqref{eq:bound}, which from Lemma~\ref{lem:saturate} only happens for the pure AdS flow~\eqref{eq:halfplaneflow}.  The corresponding boundary region~$B$ is the half-plane~$x > 0$, but it will be cleaner to conformally map the half-plane to the disk via a bulk diffeomorphism.  Working with AdS in polar-Poincar\'e coordinates
\be
ds^2 = \frac{\ell^2}{z^2} \left( dz^2 + dr^2 + r^2 d\phi^2\right),
\ee
mapping the half-plane to the disk~$r < r_0$ yields the flow
\be
\tau = \ln\left(1+\frac{(r_0^2 - r^2 - z^2)^2}{4r_0^2 z^2}\right).
\ee
Since this flow saturates the bound, it immediately follows that the renormalized entanglement entropy of a disk in vacuum is
\be
\label{eq:Sdisk}
S^\mathrm{ren}_\mathrm{disk} = -8\pi^2 c_\mathrm{eff},
\ee
which reproduces the result of~\cite{TayWoo16}.

In fact, this result could have been obtained directly from the partial result~\eqref{eq:partialbound}.  The advantage of the flow is that it allows us to calculate the weighing function~$\omega$: in the conformal frame defined by~$\Omega = z/\ell$, we obtain
\be
\omega = \lim_{z \to 0} z e^{\tau/2} = \frac{r_0^2 - r^2}{2r_0},
\ee
producing the operator
\be
\label{eq:Edisk}
\Ecal_\mathrm{disk} = 2\pi \int_{r < r_0} d^2 x \,  \frac{r_0^2 - r^2}{2r_0} \, \varepsilon.
\ee
This can immediately be recognized as the vacuum modular Hamiltonian~$H_\mathrm{disk}$ of the disk (see e.g.~\cite{BlaCas13}), as might be expected based on the saturation of the bound.

The above observation allows us to derive the first law of entanglement for perturbations of the vacuum, which we can think of as a check of the HRT formula complementary to (and more general than) those performed in~\cite{BlaCas13}.  To do so, consider a perturbation of order~$\delta$ of the CFT vacuum state.  In the bulk dual, this corresponds to a perturbation of the bulk geometry and matter fields of order~$\delta$, and therefore a perturbation to~$|^N \! K_{ab}|^2$,~$|^\Sigma \! \widehat{K}_{ab}|^2$,~$|D_a \, ^\Sigma \!K|^2$, and~$T_{ab} t^a t^b$ of order~$\delta^2$ (this latter follows because the bulk stress tensor is at least quadratic in matter fields, which vanish in vacuum).  As a result, we find that~$\dot{I}_H + I_H/2 = \Ocal(\delta^2)$, and therefore the bound is saturated to linear order~$\delta$.  Using the fact that~$\langle \varepsilon \rangle = 0$ in the vacuum, one therefore finds that
\be
\delta S^\mathrm{ren}_\mathrm{disk} = \delta \left\langle \Ecal_\mathrm{disk} \right\rangle + \Ocal(\delta^2),
\ee
which immediately reproduces the first law~\eqref{eq:firstlaw} upon replacing~$\Ecal_\mathrm{disk}$ with~$H_\mathrm{disk}$ and~$\delta S^\mathrm{ren}_\mathrm{disk}$ with~$\delta S_\mathrm{disk}$.

\subsection{A Perturbative Bound}
\label{subsec:pertbound}

While we have shown our bound reproduces the first law of entanglement (for disk-shaped regions and perturbations of the vacuum) when the bulk is described by Einstein gravity, we can obtain a much more powerful result by proceeding in the opposite direction: by exploiting positivity of the relative entropy, we can push our bound beyond the classical limit.  For simplicity, let us restrict to~$1/N$ corrections, though the discussion will apply equally well to other perturbative corrections (e.g.~in~$\alpha'$).  We may then consider a state~$\rho_B$ which is a finite-order perturbation of a classical state, by which we mean that (i) there exists a classical geometry~$g^{(0)}_{ab}$, corresponding to the infinite-$N$ limit of~$\rho_B$, which obeys Einstein's equation~\eqref{eq:Einstein}, and (ii) expectation values in the state~$\rho_B$ can be expanded as a series in~$1/N$ which truncates at some finite power.  In such a state, we therefore have that
\be
\label{eq:Cdef}
S^\mathrm{ren}(\rho_B) - \left(\left\langle \Ecal_B \right\rangle - 8\pi^2 c_\mathrm{eff} \chi_B\right) = -C^{(0)} - \sum_{i=1}^n C^{(i)},
\ee
where~$\omega$, and thus~$\Ecal_B$, is constructed from the IMC flow in the classical geometry~$g^{(0)}_{ab}$, and the~$C^{(i)}$ are organized by powers of~$1/N$ with~$C^{(0)}$ being~$\Ocal(N^2)$.  To show that the bound~\eqref{eq:bound} holds for arbitrary~$\rho_B$ (satisfying properties~(i) and~(ii) above), we must show that the right-hand side of~\eqref{eq:Cdef} is non-positive.

Now, since~$C^{(0)}$ captures the limit of a classical bulk, Theorem~\ref{thm:bound} implies that~$C^{(0)} \geq 0$.  But if~$C^{(0)} > 0$, the additional subleading terms~$C^{(i > 0)}$ can't change the sign of the right-hand side of~\eqref{eq:Cdef}, and therefore it remains negative.  Thus it remains to check the marginal case~$C^{(0)} = 0$, which by Lemma~\ref{lem:saturate} and the discussion above requires that~$g^{(0)}_{ab}$ be pure AdS and that~$B$ be a disk on flat space; thus we write~$\rho_B = \rho_\mathrm{disk}$.  This condition in particular implies that~$\rho_\mathrm{disk}$ is ``close'' to the vacuum state\footnote{By the ``vacuum state''~$\rho^{(0)}_\mathrm{disk}$, we really mean the reduced density matrix on~$B$ when~$B$ is a disk on flat space and the CFT in its vacuum state; the resulting reduced density matrix~$\rho^{(0)}_\mathrm{disk}$ will in fact be thermal with a Rindler temperature.}~$\rho^{(0)}_\mathrm{disk}$ in the sense that the energy density in the state~$\rho_\mathrm{disk}$ vanishes at order~$N^2$.  Let us therefore consider the relative entropy between the state~$\rho_\mathrm{disk}$ and the vacuum state~$\rho^{(0)}_\mathrm{disk}$, which must be non-negative:
\be
\label{eq:rhovacpositive}
S\left(\rho_\mathrm{disk} \middle| \rho^{(0)}_\mathrm{disk} \right) = \Delta \langle H_\mathrm{disk} \rangle - \Delta S_\mathrm{disk} \geq 0,
\ee
where~$H_\mathrm{disk}$ is the modular Hamiltonian of~$\rho^{(0)}_\mathrm{disk}$.  But from Section~\ref{subsec:pureAdS} above, we know that~$H_\mathrm{disk} = \Ecal_\mathrm{disk}$ (since~$g^{(0)}_{ab}$ is pure AdS); thus we obtain
\begin{subequations}
\bea
\Delta \langle H_\mathrm{disk} \rangle &= \langle \Ecal_\mathrm{disk} \rangle, \\
\Delta S_\mathrm{disk} &= \Delta S^\mathrm{ren}_\mathrm{disk} = S^\mathrm{ren}(\rho_\mathrm{disk}) + 8\pi^2 c_\mathrm{eff},
\eea
\end{subequations}
where the first line follows from the fact that the energy density vanishes to all orders in~$N$ in the state~$\rho^{(0)}_\mathrm{disk}$ and the second from equation~\eqref{eq:Sdisk} for the entropy of~$\rho^{(0)}_\mathrm{disk}$ (in this perturbative analysis, we interpret~\eqref{eq:Sdisk} as a condition that fixes the renormalization scheme).  Inserting these results back into~\eqref{eq:rhovacpositive}, we obtain
\be
S^\mathrm{ren}(\rho_\mathrm{disk}) - \left(\left\langle \Ecal_\mathrm{disk} \right\rangle - 8\pi^2 c_\mathrm{eff}\right) \leq 0;
\ee
thus we conclude that when~$C^{(0)}$ vanishes, positivity of the relative entropy requires the subleading terms~$C^{(i>0)}$ to be such that the right-hand side of~\eqref{eq:Cdef} is negative.

We have thus shown that -- thanks to positivity of the relative entropy -- the bound~\eqref{eq:bound}, which was derived purely in the limit of infinite~$N$, in fact holds perturbatively in~$1/N$ corrections to arbitrary finite order!  Moreover, the weighting function~$\omega$ is still constructed from the classical geometry, so the operator~$\Ecal_B$ is straightforward to compute from the classical holographic dual.

Note that since~$\Ecal_B = \int \omega \varepsilon$, the positivity of~$\omega$ now allows us to generalize the simple bound~\eqref{eq:partialbound} away from the vacuum (and from the classical bulk limit):
\be
\label{eq:vacbound}
\langle \varepsilon \rangle \leq 0 \mbox{ on } B \quad \Longrightarrow \quad S^\mathrm{ren}_B \leq -8\pi^2 c_\mathrm{eff} \chi_B.
\ee
This statement is nontrivial, as states with energy density that is somewhere negative are very generic~\cite{FisHic16,EpsGla65}.  It also implies that in vacuum on flat space,~$S^\mathrm{ren}_B$ is bounded above by~$-8\pi^2 c_\mathrm{eff} \chi_B$ for any region~$B$, with equality only if~$B$ is a disk on Minkowski space (since~$\langle \varepsilon \rangle$ vanishes in vacuum).  But interestingly,~\cite{FauLei15} have found that in the vacuum of \textit{any} (not just holographic) CFT on Minkowski space, disks locally maximize~$S^\mathrm{ren}_B$.  Thus our bound reproduces a result known to be true in the vacuum of any CFT.  It is therefore tempting to conjecture that, at least in vacuum, our bound should hold more generally: that is, one might conjecture that~\eqref{eq:vacbound} holds in the vacuum of \textit{any} CFT on Minkowski space.  In fact, since~\eqref{eq:vacbound} holds whenever the energy density on~$B$ is non-positive, we conjecture that it holds in any state of any CFT as long as the energy density on~$B$ is non-positive.

\section{Discussion}
\label{sec:disc}

In this Paper, we have obtained a bound on the renormalized entanglement entropy of holographic~(2+1)-dimensional CFTs in terms of a weighted energy density.  One might object that because the weighting function~$\omega$ must still be computed from the bulk geometry, we have not completely fulfilled our goal of removing the bulk from the picture entirely.  One obvious remedy is to interpret our result as the statement that \textit{there exists} a positive-definite~$\omega$ such that~\eqref{eq:bound} holds.  This alone is a nontrivial result: for instance, an immediate consequence is the bound~\eqref{eq:vacbound} whenever the energy density on~$B$ is non-positive.  But a better interpretation is that all bulk dependence in~\eqref{eq:bound} is packaged into~$\omega$, leaving an elegant relationship between two important field theory quantities: entanglement entropy and local energy density.

Indeed, our bound has a number of appealing aspects, most notably the facts that (i) it takes the form of a weighted integral over the local energy density of the CFT, which is a well-understood local operator; (ii) it holds perturbatively in~$1/N$ corrections to arbitrary finite order, even though the bound operator~$\Ecal_B$ is calculated only from the infinite-$N$ limit; and (iii) it is consistent with known results regarding the entanglement entropy of the vacuum state of more general non-holographic CFTs on Minkowski space.  In fact, a special case of our bound reproduces the result of~\cite{AleMaz08}, which is a global generalization of the local statements of~\cite{AllMez14,Mez14,FauLei15}: disks maximize the renormalized entanglement entropy of CFTs in the vacuum state.  Because of its consistency with these known CFT results, we have conjectured (subject to certain constraints) that in any state of a CFT, the renormalized entanglement entropy of any region~$B$ with non-positive energy density is bounded above by~$-8\pi^2 c_\mathrm{eff} \chi_B$ as long as~$B$ lies on a moment of time symmetry (here~$c_\mathrm{eff}$ is interpreted as a measure of the degrees of freedom in the CFT, computed as usual from the coefficient of the two-point function of the stress tensor).  Any scheme-dependence in the definition of the renormalized entropy can be removed by requiring that for disk-shaped regions, the entanglement entropy equal precisely~$-8\pi^2 c_\mathrm{eff}$.

Nevertheless, our derivation via the inverse mean curvature flow required making some simplifying assumptions.  Let us therefore list the limitations of~\eqref{eq:bound} and possible ways to address them, as well as interesting open questions that we leave to future work.

\textbf{Weak IMC Flows.} The bound in Theorem~\ref{thm:bound} follows from a monotonicity property of the reduced Hawking mass along a smooth IMC flow.  But as discussed in Section~\ref{sec:IMCF}, IMC flows need not always be smooth, and in general one must instead consider weak IMC flows.  For flows of surfaces without boundary, the results of~\cite{HuiIlm01,LeeNev15} show that it is possible to define a Hawking mass functional which is monotonic even on these weak flows.  Combined with the fact that weak IMC flows of compact surfaces with boundary have been shown to exist~\cite{Mar15}, we think it quite likely that there exists a way of preserving the monotonicity property even along weak boundary-anchored flows.  However, this remains to be proven.  Relatedly, a detailed analysis of such weak flows might provide more insights into the behavior of the weighting function~$\omega$, which would in turn put stronger constraints on the behavior of the operator~$\Ecal_B$.

\textbf{Higher Dimensions.}  Our bound was derived only for a~$(2+1)$-dimensional CFT, and therefore a natural question is whether it may be extended to higher dimensions.  For historical reasons, the original considerations of the IMC flow focused exclusively on~$(3+1)$-dimensional spacetimes.  However, the IMC flow has since been generalized to arbitrary dimensions; it therefore remains to find a Hawking mass-like functional that remains monotonic even in higher dimensions; whose value on boundary-anchored extremal surfaces recovers their renormalized area; and whose asymptotic behavior can be expressed in terms of the CFT stress tensor.  One way to potentially construct such an operator is to generalize the special properties that~$I_H$ has in~$(3+1)$ bulk dimensions.  For instance, at least in even bulk dimensions, one could ensure that~$I_H$ contains an integral over an Euler density, and would need to engineer other geometric contributions to~$I_H$ to precisely cancel the counterterms that appear in the boundary stress tensor~\eqref{eq:BalaKraus} in higher dimensions.  At least superficially, it appears these requirements are consistent with one another: in~$2(n+1)$ bulk dimensions both the Euler density and the counterterms that appear in the boundary stress tensor contain~$n$ powers of the intrinsic curvature.  It is less clear, however, how to proceed in odd bulk dimensions; we leave this direction to future work.

\textbf{Deforming the CFT.} Perhaps the strongest assumption necessary for~\eqref{eq:bound} to hold is that the bulk stress tensor decay like~$o(\Omega^3)$, which is necessary to ensure that the asymptotic value of the Hawking mass be finite.  In field theoretic terms, this falloff requirement implies that except for the stress tensor, all CFT operators with conformal dimension~$\Delta > 3/2$ have a vanishing source (and that operators with~$\Delta < 3/2$ have a vanishing one-point function).  However, we expect that it should be possible to generalize the definition of the reduced Hawking mass to allow for more nontrivial deformations of the CFT: if the bulk stress tensor falls of more slowly than~$o(\Omega^3)$, the boundary stress tensor~$\T_{ab}$ receives additional contributions from matter counterterms.  Then the reduced Hawking mass can be modified by matter terms chosen to precisely cancel these additional contributions to~$\T_{ab}$.  However, it is unclear whether or not such a modification can be performed while preserving the crucial monotonicity property.  If the bound can be generalized in this way, however, then by deforming the CFT one could explore RG flow of renormalized entanglement entropy.  Indeed, in such a case, it would potentially be possible to give the IMC flow itself an interpretation as a geometrization of an RG flow.

\textbf{Connections to Complexity.} Recently, it has been conjectured that the complexity of a holographic CFT state is dual to either the volume of a maximal-time slice of the bulk or the spacetime volume of the Wheeler-deWitt patch of the bulk~\cite{Sus14,Sus14b,StaSus14,BroRob15,BroRob15b,BroSus16,BenCar16}.  While we note that the latter proposal currently seems to be favored in the literature, it is worth understanding in more detail the prominent role of the maximal-volume time slice~$N$ used in constructing our bound.  To that end, note that the bound~\eqref{eq:bound} essentially comes about from the equality
\be
\label{eq:complexbound}
\lim_{\tau \to \infty} e^{\tau/2} I_H[\Sigma_\tau] = I_H[\Sigma_0] + \int_0^\infty e^{\tau/2} \left(\dot{I}_H[\Sigma_\tau] + \frac{1}{2} I_H[\Sigma_\tau]\right) \, d\tau,
\ee
followed by the observation that~$\dot{I}_H + I_H/2 \geq 0$.  But since~$\dot{I}_H + I_H/2$ can be written as an integral over slices of the time slice~$N$, the integral in~\eqref{eq:complexbound} can be re-expressed as an integral over the \textit{entire} time slice~$N$.  One should therefore be able to bound that term by an expression proportional to the volume of~$N$ to obtain a refined bound between~$S^\mathrm{ren}_B$ and~$\Ecal_B$.  If, as per the early holographic complexity conjectures, we interpret the volume of~$N$ as the complexity of the CFT state on~$B$, then this refined bound would also have an interpretation in field theoretic terms, and would nicely tie together entanglement entropy, local energy, and complexity.

\textbf{Quantum Energy Inequalities.} Here, we have interpreted the bound~\eqref{eq:bound} as a bound on~$S^\mathrm{ren}_B$ in terms of the weighted energy density~$\Ecal_B$.  But one could equally well interpret~\eqref{eq:bound} as a bound on the expectation value of~$\Ecal_B$ in terms of~$S^\mathrm{ren}_B$.  In this sense, our bound can be interpreted as a type of quantum energy inequality (QEI); that is, a lower bound on the weighted energy density of a QFT (see e.g.~\cite{Few12} for a review).  However, note that in our case, the weighting function~$\omega$ is fixed by the classical bulk geometry, and therefore by the CFT state.  In QEIs, on the other hand, one typically desires the freedom to choose the weighting function arbitrarily.  Nevertheless, deriving QEIs even for free fields is quite nontrivial, and therefore our bound can be thought of as a first step towards deriving QEIs in holographic CFTs.
 
\textbf{A Field Theoretic Derivation.} While our bound was derived only for states of holographic CFTs with Einstein gravity duals, in Section~\ref{sec:perturbative} we showed that the bound in fact continues to hold under perturbative quantum corrections.  It would therefore be interesting to investigate whether a bound of the form~\eqref{eq:bound} can be given a purely field theoretic derivation, without the need for holography at all.  That is, can it be shown that for any region~$B$ of any CFT, the renormalized entropy density is bounded from below by an appropriately weighted local energy density of~$B$?  When the energy density everywhere in~$B$ has the same sign as~$S^\mathrm{ren}_B$, the answer is trivially ``yes'': one could simply take~$\omega = S^\mathrm{ren}_B/(\langle \varepsilon \rangle \mathrm{Vol}(B))$ so that trivially~$\langle \Ecal_B \rangle = S^\mathrm{ren}_B$.  The question becomes less trivial when~$S^\mathrm{ren}_B$ and~$\langle \varepsilon \rangle$ have different signs somewhere in~$B$.

\section*{Acknowledgements}

It is a pleasure to thank Francesco Aprile, Oscar Dias, Netta Engelhardt, Piermarco Fonda, Simon Gentle, Don Marolf, Nick Poovuttikul, and Kostas Skenderis for discussions and correspondence.  We also thank the two anonyous referees for their useful comments on an earlier draft of this manuscript.  SF was supported by the ERC Advanced grant No.~290456 and STFC grant ST/L00044X/1.

\appendix

\section{Conformal Transformation of Extrinsic Curvatures}
\label{app:Kconformal}

In this Appendix we useful results on the conformal transformation properties of extrinsic and geodesic curvatures.  Let us therefore consider a hypersurface~$N$ of dimension~$n$ embedded in a spacetime~$(M,g_{ab})$ (here we allow~$N$ and~$M$ to have arbitrary dimension and signature, as long as neither is degenerate).  Define the extrinsic curvature of~$N$ to be
\be
K^a_{bc} = -{h_b}^d {h_c}^e \grad_d {h_e}^a,
\ee
where~$h_{ab}$ is the induced metric on~$N$ and~$\grad_a$ is the covariant derivative operator compatible with~$g_{ab}$.  We also remind the reader that~$K^a_{bc}$ only has components orthogonal to~$N$ in its upper index and tangent to~$N$ in its lower two indices.

Under a conformal rescaling~$\tilde{g}_{ab} = \Omega^2 g_{ab}$, the induced metric on~$N$ transforms as~$\tilde{h}_{ab} = \Omega^2 h_{ab}$, and therefore its extrinsic curvature is
\be
\widetilde{K}^a_{bc} = -\tilde{h}_b^{\phantom{b}d} \tilde{h}_c^{\phantom{c}e} \widetilde{\grad}_d \tilde{h}_e^{\phantom{e}a} = -{h_b}^d {h_c}^e \grad_d {h_e}^a - {h_b}^d {h_c}^e \left(C^a_{df} {h_e}^f - C^f_{de} {h_f}^a \right),
\ee
where~$\widetilde{\grad}_a$ is the covariant derivative operator compatible with~$\tilde{g}_{ab}$, the tensor~$C^a_{bc}$ is the difference~$\widetilde{\grad}_a - \grad_a$ between these two connections, and we recognized that~$\tilde{h}_a^{\phantom{a}b} = {h_a}^b$.  In terms of~$\Omega$,~$C^a_{bc}$ is given by (see e.g.~Appendix D of~\cite{Wald})
\be
C^a_{bc} = 2{\delta^a}_{(b} \grad_{c)} \ln \Omega - g_{bc} \grad^a \ln \Omega,
\ee
so that we find that the extrinsic curvature transforms as
\be
\label{eq:Kabcconformal}
\widetilde{K}^a_{bc} = K^a_{bc} + h_{bc} \left(g^{ad} - h^{ad}\right) \grad_d \ln \Omega.
\ee
It is then straightforward to see that the trace-free extrinsic curvature transforms homogeneously under conformal transformations:
\be
\widetilde{K}^a_{bc} - \frac{1}{n} \widetilde{K}^a \tilde{h}_{bc} = K^a_{bc} - \frac{1}{n} K^a h_{bc},
\ee
where~$K^a \equiv h^{bc} K^a_{bc}$ and likewise for~$\widetilde{K}^a$.  

Consider now two special cases.  First, if~$N$ is codimension-one with unit normal~$n^a$, the upper index of~$K^a_{bc}$ only has a component in the single normal direction~$n^a$.  It is thus customary to take the extrinsic curvature to be~$K_{ab} = n_c K^c_{ab}$\footnote{Note that this definition is equivalent to the expression~$K_{ab} = {h^c}_a {h^d}_b \grad_c n_d$ more commonly seen in the literature.  By this convention, the mean curvature of the sphere in~$\mathbb{R}^3$ is positive for outward-pointing~$n^a$.}, in which case we find (using~$\tilde{n}_a = \Omega \, n_a$) that the corresponding trace-free extrinsic curvature transforms as
\be
\label{eq:Kconformal}
\widetilde{K}_{ab} - \frac{1}{n} \widetilde{K} \tilde{h}_{ab} = \Omega \left(K_{ab} - \frac{1}{n} K h_{ab}\right).
\ee

Next, let~$N$ be a (one-dimensional) curve of signature~$\eps$ with unit tangent~$u^a$ (that is,~$\eps = u^2 = \pm 1$).  Then the lower two indices of~$K^a_{bc}$ have components only in the direction~$u^a$, and the trace~$K^a = K^a_{bc} u^b u^c$ is just the negative geodesic acceleration of~$N$:~$a^a = -K^a$.  From~\eqref{eq:Kabcconformal} (and using~$\tilde{u}^a = \Omega^{-1} \, u^a$ and~$h_{ab} = \eps u_a u_b$) we find~$a^a$ transforms as
\be
\tilde{a}^a = \frac{1}{\Omega^2}\left[a^a - \left(\eps g^{ab} - u^a u^b\right) \grad_b \ln \Omega\right].
\ee
If~$M$ is two-dimensional, let~$n^a$ be the unit normal vector to~$N$ (which for nonzero~$a^a$ is given by~$n^a = a^a/|a|$); then the geodesic curvature of~$N$ is~$k_g = n_a a^a$, which thus transforms like
\be
\label{eq:kgconformal}
\tilde{k}_g = \frac{1}{\Omega} \left(k_g - \eps n^a \grad_a \ln\Omega\right).
\ee

Finally, we may use these conformal transformation properties to derive useful properties of extremal surfaces in asymptotically locally AdS spacetimes.  The first is the well-known property that extremal surfaces intersect the boundary orthogonally; the second constrains the asymptotic falloff of the extrinsic curvature; the third reproduces for convenience a result of~\cite{AleMaz08}.

\begin{lem}
\label{lem:extremalorthog}
Any non-null extremal surface which intersects the conformal boundary of an asymptotically locally AdS spacetime must be orthogonal to it.
\end{lem}

\begin{proof}
Let~$N$ be a non-null extremal surface that intersects the conformal boundary~$\partial M$, and call the induced metric and extrinsic curvature of~$N$~$h_{ab}$ and~$K^a_{bc}$, respectively.  Consider a defining function~$\Omega$ that conformally compactifies the geometry, so that the rescaled metric~$\tilde{g}_{ab} = \Omega^2 g_{ab}$ is regular at the conformal boundary~$\partial M$.  By~\eqref{eq:Kabcconformal}, we find that
\be
\widetilde{K}^a \equiv \tilde{h}^{bc} \widetilde{K}^a_{bc} = \frac{n}{\Omega} \left(\tilde{g}^{ab} - \tilde{h}^{ab}\right)\gradt_b \Omega,
\ee
where~$n = \mathrm{dim}(N)$ and we exploited the extremality condition~$K^a \equiv h^{bc} K^a_{bc} = 0$.  Now, since~$N$ must be differentiable at~$\partial M$\footnote{This follows from the fact that since it is extremal,~$N$ obeys local equations determined by~$\tilde{g}_{ab}$ and~$\Omega$, which are both differentiable at the boundary.},~$\widetilde{K}^a$ must be regular there; since~$\Omega|_{\partial M} = 0$, this implies that
\be
\left(\tilde{g}^{ab} - \tilde{h}^{ab}\right)\gradt_b \Omega|_{\partial M} = 0.
\ee
But~$\gradt_b \Omega|_{\partial M} \propto \tilde{n}_b$, where~$\tilde{n}^a$ is the normal to the AdS boundary.  Moreover, note that the tensor~$(\tilde{g}^{ab} - \tilde{h}^{ab})$ is just the projector onto the space orthogonal to~$N$.  Thus we find that at the boundary,~$\tilde{n}^a$ has no orthogonal component to~$N$, implying that~$N$ is normal to the boundary.
\end{proof}

\begin{lem}
\label{lem:Kfalloff}
Let~$N$ be a non-null extremal surface which intersects the conformal boundary of an asymptotically locally AdS spacetime.  If~$\partial N \subset \partial M$ is a conformally totally geodesic surface (that is, if there exists a conformal frame in which the projection of~$^{\partial N} \! \widetilde{K}^a_{bc}$ to~$\partial M$ vanishes), then the components of the extrinsic curvature~$K^a_{bc}$ of~$N$ in any orthonormal frame have asymptotic falloff
\be
K^{\hat{\mu}}_{\hat{\nu}\hat{\rho}} = \Ocal(\Omega^2)
\ee
for any defining function~$\Omega$.
\end{lem}

\begin{proof}
If~$\mathrm{dim}(N) = 1$, then~$N$ is a geodesic and its extrinsic curvature vanishes identically; hence the falloff is obeyed trivially.  Let us therefore consider~$\mathrm{dim}(N) \geq 2$.

Consider a defining function~$\Omega$ that conformally compactifies the geometry and with respect to which~$\partial N$ is a totally geodesic surface; that is,~$\tilde{g}_{ab} = \Omega^2 g_{ab}$ is regular at~$\partial M$ and~$\partial N$ is a totally geodesic surface in~$\partial M$.  Next, on~$\partial N$ consider an arbitrary unit\footnote{Normalized vector fields dressed with tildes are normalized with respect to~$\tilde{g}_{ab}$.} vector field~$\tilde{n}^a$ normal to~$\partial N$ and tangent to~$\partial M$, and an additional arbitrary unit vector field~$\tilde{t}^a$ tangent to~$\partial N$.  By Lemma~\ref{lem:extremalorthog},~$N$ is orthogonal to~$\partial M$, and therefore~$\tilde{n}^a$ is also normal to~$N$.  We may thus smoothly extended~$\tilde{n}^a$,~$\tilde{t}^a$ to unit vector fields on~$N$ in a neighborhood of~$\partial M$ by requiring that they remain orthogonal and tangent to~$N$, respectively.  (This extension is highly non-unique; it can be performed, for instance, by parallel transport along an arbitrary family of curves on~$N$.)  Thus they give rise to unit vector fields~$n^a = \Omega \tilde{n}^a$,~$t^a = \Omega \tilde{t}^a$ with respect to~$g_{ab}$ on~$N$ in a neighborhood of~$\partial M$.

Next, consider the components~$\tilde{n}_a \tilde{t}^b \tilde{t}^c \widetilde{K}^a_{bc}$ of the extrinsic curvature of~$N$ with respect to~$\tilde{g}_{ab}$; these are related to the extrinsic curvature of~$N$ with respect to~$g_{ab}$ by~\eqref{eq:Kabcconformal}:
\be
\label{eq:Kcomponenttransform}
\tilde{n}_a \tilde{t}^b \tilde{t}^c \widetilde{K}^a_{bc} = \frac{1}{\Omega}\left(n_a t^b t^c K^a_{bc} + \eps n^a \grad_a \ln \Omega\right),
\ee
where~$\eps \equiv t^2 = \pm 1$, depending on the signature of~$t^a$.  Now, on~$\partial N$ we have that
\be
\tilde{n}_a \tilde{t}^b \tilde{t}^c \widetilde{K}^a_{bc}|_{\partial N} = \tilde{n}_a \tilde{t}^b \tilde{t}^c \, ^{\partial N} \! \widetilde{K}^a_{bc} = 0,
\ee
where the second equality is a consequence of the fact that~$\partial N$ is a totally geodesic surface with respect to~$\tilde{g}_{ab}$.  This implies that near the boundary, we must have~$\tilde{n}_a \tilde{t}^b \tilde{t}^c \widetilde{K}^a_{bc} = \Ocal(\Omega)$.  Comparing with~\eqref{eq:Kcomponenttransform}, this in turn implies that
\be
n_a t^b t^c K^a_{bc} + \eps n^a \grad_a \ln \Omega = \Ocal(\Omega^2).
\ee
But since~$t^a$ is arbitrary, this implies that
\be
n_a K^a_{bc} = -h_{bc} \, n^a \grad_a \ln\Omega + \Ocal(\Omega^2),
\ee
where the~$\Ocal(\Omega^2)$ stands for a tensor whose orthonormal components are~$\Ocal(\Omega^2)$.  Finally, because~$N$ is extremal, we have that
\be
0 = n_a K^a = -n \, n^a \grad_a \ln \Omega + \Ocal(\Omega^2),
\ee
where as before,~$n = \mathrm{dim}(N)$.  Therefore 
\be
n^a \grad_a \ln \Omega = \Ocal(\Omega^2),
\ee
and so we conclude that in fact
\be
n_a K^a_{bc} = \Ocal(\Omega^2)
\ee
for this choice of defining function~$\Omega$.  But since~$n^a$ is an arbitrary unit vector normal to~$N$ and since any two defining functions are related by an~$\Ocal(1)$ factor, this completes the proof.
\end{proof}

\begin{lem}
\label{lem:geocurvature}
Let~$\Sigma$ be a two-dimensional spacelike boundary-anchored extremal surface in an asymptotically locally AdS spacetime~$(M,g_{ab})$ with AdS scale~$\ell$.  As described in Section~\ref{subsec:renarea}, let~$\Sigma^\eps$ be the corresponding regulated surface cut off at~$Z = \eps/\ell$, where~$Z$ is a Fefferman-Graham defining function.  Then
\be
\label{eq:kgexpansion}
\int_{\partial \Sigma^\eps} k_g = -\frac{L}{\eps} + \Ocal(\eps),
\ee
where~$k_g$ is the geodesic curvature of~$\partial \Sigma^\eps$ in~$\Sigma$ with respect to the outward-pointing normal and~$L$ is the length of~$\partial \Sigma$ in the conformal frame defined by~$Z$ (and the natural volume element on~$\partial \Sigma^\eps$ is understood).
\end{lem}

\begin{proof}
The proof reproduces elements of that in~\cite{AleMaz08}.  We use the conformal transformation~\eqref{eq:kgconformal} of the geodesic curvature~$k_g$ to express it in terms of the geodesic curvature~$\tilde{k}_g$ of~$\partial \Sigma^\eps$ in the compactified geometry~$\tilde{g}_{ab} = Z^2 g_{ab}$:
\be
k_g = Z \left( \tilde{k}_g + \tilde{n}^a \gradt_a \ln Z\right) = \left( \tilde{n}^a \gradt_a Z + \frac{\eps}{\ell} \, \tilde{k}_g \right),
\ee
where~$\tilde{n}^a$ is the unit outward-pointing normal to~$\partial \Sigma^\eps$ in the compactified geometry and we have used the fact that on~$\partial \Sigma^\eps$,~$Z = \eps/\ell$.  With respect to~$\tilde{g}_{ab}$, the conformal boundary~$\partial M$ is a totally geodesic surface; i.e.,~any curve on it has no geodesic acceleration in the direction normal to the boundary.  But~$\tilde{k}_g$ is the component of the geodesic acceleration of~$\partial \Sigma^\eps$ tangent to~$\Sigma$, so since~$\Sigma$ intersects the boundary orthogonally (since it is an extremal surface), we have that~$\tilde{k}_g = \Ocal(\eps)$.  Also because~$\Sigma$ is orthogonal to the boundary, we have that~$\tilde{n}_a = -\gradt_a Z/ |\gradt Z| + \Ocal(\eps^2)$, and thus
\be
\tilde{n}^a \gradt_a Z = -|\gradt Z| + \Ocal(\eps^2) = -1/\ell + \Ocal(\eps^2),
\ee
where in the second equality we used the property~\eqref{eq:FGZ}.  We therefore find that~$k_g = -1/\ell + \Ocal(\eps^2)$.  Using the fact that the proper distance transforms as~$ds = Z^{-1} \tilde{ds} = (\ell/\eps) \tilde{ds}$, we thus find that
\be
\label{eq:kgexpansion}
\int_{\partial \Sigma^\eps} k_g \, ds = -\frac{L}{\eps} + \Ocal(\eps),
\ee
which completes the proof.
\end{proof}

\section{Regularity of the Reduced Hawking Mass}
\label{app:massreg}

In this Appendix, we show that the reduced Hawking mass~$I_H[\Sigma]$ is finite for any boundary-anchored surface~$\Sigma$, as long as the bulk stress tensor~$T_{ab}$ falls off sufficiently rapidly at infinity.  To see this, first note that by using the Gauss-Codazzi equations and the Einstein equation~\eqref{eq:Einstein}, the Ricci curvature of~$\Sigma$ can be decomposed in terms of the full stress tensor, the Weyl tensor of~$M$, and the extrinsic curvature of~$N$ and~$\Sigma$.  Using the fact that~$N$ is a maximal-volume surface (so~$^N \! K = 0$), the result is
\begin{multline}
\label{eq:sigmaricci}
^\Sigma \! R = -\frac{2}{\ell^2} - 2C_{abcd} n^a t^b n^c t^d + \, ^\Sigma \! K^2 - \, ^\Sigma \! K_{ab} \, ^\Sigma \! K^{ab} \\ + \left| ^N \! K_{ab} \right|^2 - 2\left| ^N \! K_{ab} n^b \right|^2 + 8\pi G_N \left(\sigma^{ab} T_{ab} - \frac{2}{3} T\right),
\end{multline}
where recall that~$t^a$ and~$n^a$ are the unit normals to~$N$ and~$\Sigma$ (in~$N$), respectively.  Plugging the above expression into~\eqref{eq:redHawkingmass}, we obtain
\begin{multline}
\label{eq:integral}
I_H[\Sigma] = 2 \int_{\Sigma} \left(\left| ^N \! \widehat{K}_{ab} \right|^2 - 2\left| ^N \! \widehat{K}_{ab} n^b \right|^2 - \left|^\Sigma \! \widehat{K}_{ab}\right|^2 \right) \\ - 4\int_{\Sigma} C_{abcd} n^a t^b n^c t^d + 16\pi G_N \int_{\Sigma} \left(\sigma^{ab} T_{ab} - \frac{2}{3} T\right),
\end{multline}
where we have rewritten everything in terms of the trace-free extrinsic curvatures
\begin{subequations}
\bea
^N \! \widehat{K}_{ab} &\equiv \, ^N \! K_{ab} - \frac{1}{3} \, ^N \! K h_{ab} = \, ^N \! K_{ab}, \\
^\Sigma \! \widehat{K}_{ab} &\equiv \, ^\Sigma \! K_{ab} - \frac{1}{2} \, ^\Sigma \! K \sigma_{ab},
\eea
\end{subequations}
with the first equality holding because~$^N \! K = 0$.  As an aside, note that when~$\Sigma$ is extremal, inserting~\eqref{eq:integral} into equation~\eqref{eq:ArenI} produces an expression for~$\Acal[\Sigma]$ analogous to ones appearing in~\cite{FonSem15}.

Now, consider a conformal frame~$\tilde{g}_{ab} = \Omega^2 g_{ab}$ which is regular at~$\partial \Sigma \subset \partial M$ (such a frame must exist, by definition~\ref{def:bndryanchored}).  From~\eqref{eq:Kconformal} (and from the conformal transformation properties~$\tilde{t}^a = \Omega^{-1} t^a$,~$\tilde{n}^a = \Omega^{-1} n^a$), it is straightforward to see that the integral in~\eqref{eq:integral} containing the trace-free extrinsic curvatures is a conformal invariant.  Likewise, the integral containing the Weyl scalar is a conformal invariant as well.  Thus each of those integrals can be evaluated in the conformally compactified geometry~$\tilde{g}_{ab}$, and are thus clearly finite.  Similarly, the integral containing the stress tensor will be finite as well as long as the stress tensor has the asymptotic falloff
\be
\sigma^{ab} T_{ab} - \frac{2}{3} T = \Ocal(\Omega^2).
\ee
This condition will hold if the components of the stress tensor in any orthonormal frame have an asymptotic falloff of~$T_{\hat{\mu} \hat{\nu}} = \Ocal(\Omega^2)$.  Thus under this assumption, we find that~$I_H[\Sigma]$ is finite (and generically nonzero) when evaluated on any boundary-anchored surface~$\Sigma$.

In the special case where the bulk is pure AdS and~$\Sigma$ lies on a bulk moment of time symmetry, we have that~$^N \! K_{ab}$,~$C_{abcd}$, and~$T_{ab}$ all vanish, in which case~\eqref{eq:integral} reproduces the expression~\eqref{eq:IAdS} claimed in the main text.

\section{Time Evolution of Integrals}
\label{app:intevolution}

In this Appendix, we review generally how to compute the time evolution of geometric integrals.  Our approach here is a more formal version of the so-called calculus of moving surfaces~\cite{Gri13}.

Consider a one-parameter family of~$n$-dimensional surfaces~$S_\tau$ in a Riemannian manifold (these need not be IMC flow surfaces, and they may have arbitrary codimension).  Let the induced metric, natural volume form, and extrinsic curvature of the~$S_\tau$ be~$\sigma_{ab}$,~$\eps_{a_1 \cdots a_n}$, and~$K^a_{bc}$, respectively\footnote{Our definition of~$K^a_{bc}$ is given in Appendix~\ref{app:Kconformal}, and recall that the natural volume form has the property that~$\eps_{a_1 \cdots a_n} \eps^{a_1 \cdots a_n} = n!$, and thus that~$D_b \eps_{a_1 \cdots a_n} = 0$, with~$D_a$ the covariant derivative on the~$S_\tau$.  See e.g.~Appendix B of~\cite{Wald} for details and conventions.}.

Next, let us introduce a flow field~$\tau^a$ whose integral curves map the~$S_\tau$ to one another and which is normalized by the condition~$\tau^a \grad_a \tau = 1$ (where we interpret~$\tau$ as a scalar over the family of surfaces).  Then the rate of change of any geometric object on the~$S_\tau$ with respect to~$\tau$ is given by a Lie derivative along~$\tau^a$.  In particular, consider an arbitrary integral
\be
F = \int_S \mathbf{f}
\ee
for some~$n$-form~$\mathbf{f}$ on the~$S_\tau$.  The~$\tau$-derivative of~$I$ is then given by
\be
\dot{F} = \int_S \pounds_\tau \mathbf{f},
\ee
which using Cartan's formula and Stokes' theorem can be expressed as
\be
\dot{F} = \int_S \left(\iota_\tau d \mathbf{f} + d(\iota_\tau \mathbf{f})\right) = \int_S \iota_\tau d\mathbf{f} + \int_{\partial S} \iota_\tau \mathbf{f}.
\ee
Now,~$\iota_\tau d\mathbf{f}$ is an~$n$-form, and therefore its restriction to the tangent space of~$S_\tau$ must be proportional to~$\epsb$.  We thus write~$\iota_\tau d\mathbf{f} = h \epsb + \cdots$ for some scalar~$h$, where~$\cdots$ are terms whose projection onto~$S_\tau$ vanishes.  Contracting both sides with~$\eps^{a_1 \cdots a_n}$ and writing~$\mathbf{f} = f \epsb$ for a scalar~$f$, we find after some manipulation that
\bea
h = \tau^b_\perp \grad_b f + (\tau_\perp)_b K^b f,
\eea
where~$K^a \equiv \sigma^{bc} K^a_{bc}$ and where~$\tau^a_\perp$ is the component of~$\tau^a$ normal to the~$S_\tau$:~$\tau^a_\perp \equiv \tau^a - {\sigma^a}_b \tau^b$.

Likewise, the boundary term in~$\dot{F}$ can be simplified:
\begin{subequations}
\bea
(\iota_\tau f \eps)_{a_1 \cdots a_{n-1}} &= f \tau^b \eps_{b a_1 \cdots a_{n-1}} = f \tau^b_\parallel \eps_{b a_1 \cdots a_{n-1}} \\
					&= f (v_b \tau^b_\parallel) \, ^\partial \! \eps_{a_1 \cdots a_{n-1}},
\eea
\end{subequations}
where~$v^a$ is the outward-pointing unit normal to~$\partial S$ in~$S$,~$^\partial \! \eps_{a_1 \cdots a_{n-1}}$ is the natural volume form on~$\partial S$, and~$\tau^a_\parallel$ is the component of~$\tau^a$ tangent to~$S_\tau$:~$\tau^a_\parallel \equiv {\sigma^a}_b \tau^b$.  Putting these results together, we find
\be
\label{eq:Fdot}
\dot{F} = \int_S \left(\tau^a_\perp \grad_a f + (\tau_\perp)_a K^a f \right) + \int_{\partial S} v_a  \tau^a_\parallel f,
\ee
where the natural volume form in each integral is understood.  Roughly speaking, the first term can be interpreted as the contribution to~$\dot{I}$ from a variation in~$f$ in the~$n^a$ direction, the second as a contribution from the expansion or contraction of the~$S_\tau$ as they flow in the~$n^a$ direction, and the third as a contribution from any current through the boundary of~$S_\tau$.

For the special case of IMC flow surfaces, we have that the~$S_\tau$ are codimension-one in~$N$, and thus it is natural to write their mean curvature as~$K = n_a K^a$, where~$n^a$ is the unit normal to them.  The IMC flow condition requires that~$\tau^a_\perp = n^a/K$, and thus we obtain
\be
\dot{F} = \int_S \left(\dot{f} + f\right) + \int_{\partial S} f v_a \tau^a_\parallel,
\ee
where~$\dot{f} \equiv (n^a/K) \grad_a f$.  In the case where the~$S$ are compact and without boundary, the boundary term above vanishes, and we recover equation~\eqref{eq:Fdotcompact} quoted in the main text.

\section{IMC Flows in Pure AdS}
\label{app:pureAdS}

In this Appendix we derive the half-plane pure AdS IMC flow~\eqref{eq:halfplaneflow}.  Working in the Poincar\'e coordinates of equation~\eqref{eq:PoincareAdS}, the RT surface for the half-plane~$z=0$,~$x \geq 0$ is the half-plane~$x = 0$,~$z \geq 0$.  By symmetry, the solution to the IMC flow equation~\eqref{eq:flowpde} will be a function of~$x$ and~$z$ only:~$\tau = \tau(x,z)$.

It is difficult to solve~\eqref{eq:flow} exactly in full generality, so we look for perturbative solutions about the RT surface.  Specifically, we will look for a series solution in~$x$:
\be
\tau(x,z) = \sum_{n=1}^\infty \tau_n(z) x^n.
\ee
Since we expect the flow surfaces to be anchored on the line~$x = 0$,~$z = 0$,~$\tau$ must be singular there, and therefore the coefficients~$\tau_n(z)$ must be singular at~$z = 0$\footnote{This is just an order of limits statement: at fixed nonzero~$x$,~$\tau$ should diverge as~$z \to 0$; at fixed nonzero~$z$,~$\tau$ should vanish as~$x \to 0$.}.

Working order-by-order in~$x$, we find that~$\tau_1(z) = 0$, and that~$\tau_2(z)$ must obey
\be
z\left[3\tau_2'(z)^2 - 2 \tau_2(z) \tau_2''(z)\right] +4\tau_2(z) \tau_2'(z) + 8z \tau_2(z)^3 = 0.
\ee
Searching for a singular solution to the above via a Frobenius series, we find the particular solution~$\tau_2(z) = 1/z^2$.  Continuing to solve~\eqref{eq:flowpde} order-by-order, it is then possible to obtain the series solution
\be
\tau(x,z) = \frac{x^2}{z^2} - \frac{1}{2} \frac{x^4}{z^4} + \frac{1}{3} \frac{x^6}{z^6} - \frac{1}{4} \frac{x^8}{z^8} + \cdots,
\ee
which can be resummed into
\be
\tau(x,z) = \ln\left(1+\frac{x^2}{z^2}\right),
\ee
which solves~\eqref{eq:flowpde} exactly.  This is precisely the solution~\eqref{eq:halfplaneflow} invoked in the main text.

\bibliographystyle{jhep}
\bibliography{all}

\end{document}